\documentclass[prx,floatfix,longbibliography,twocolumn,superscriptaddress]{revtex4-2}

\usepackage[pdftex]{graphicx}
\usepackage{amsmath,amsfonts,amsbsy,amssymb,amsthm,mathtools}
\usepackage{mathrsfs}
\usepackage{epstopdf}
\usepackage{ulem,color}
\usepackage{enumitem}
\usepackage{siunitx}
\usepackage{amssymb,amsmath,stmaryrd}
\usepackage{ulem,color}
\usepackage{MnSymbol}
\usepackage{environ}
\usepackage{xcolor}
\usepackage{algpseudocode}
\usepackage{comment}
\usepackage{appendix}
\usepackage{float}

\usepackage{tikz}
\usetikzlibrary{quantikz}
\usetikzlibrary{automata,arrows,positioning, calc}

\usepackage{newfloat}

\DeclareFloatingEnvironment[
    fileext=loa,
    listname=List of Algorithms,
    name=ALGORITHM,
]{algorithm}

\AtEndEnvironment{algorithm*}{\noindent\hrulefill\par\nobreak\vskip-5pt}

\usepackage{pgfplots}
\usepackage{pstool}
\usepackage{tikzsymbols}
\usetikzlibrary{plotmarks,spy}

\newtheorem{lemma}{Lemma}

\newtheorem{theorem}{Theorem}

\newtheorem{fact}[theorem]{Fact}

\theoremstyle{definition}
\newtheorem{definition}{Definition}

\newcommand{\Cmat}{\mathbf{C}}
\newcommand{\Dmat}{\mathbf{D}}

\definecolor{KB}{rgb}{0.4,0.3,0.9}

\usepackage[colorlinks = true, citecolor=blue, linkcolor=blue, urlcolor=blue]{hyperref}
\usepackage[T1]{fontenc}
\usepackage{mwe}
\usepackage{tikz}
\usetikzlibrary{quantikz}

\newcommand{\ketbra}[1]{|{#1}\>\mkern-4mu\<{#1}|}
\newcommand{\tr}{\textup{Tr}}
\renewcommand{\>}{\rangle}
\newcommand{\<}{\langle}
\newcommand{\C}{{\mathbb{C}}} %complex numbers
\renewcommand{\ll}{\llangle}
\newcommand{\rr}{\rrangle}

\newcommand{\CQT}{Centre for Quantum Technologies, National University of Singapore, 3 Science Drive 2, Singapore 117543.\looseness=-1}
\newcommand{\NTU}{Nanyang Quantum Hub, School of Physical and Mathematical Sciences, Nanyang Technological University, Singapore 639673.\looseness=-1}
\newcommand{\ihpc}{Institute of High Performance Computing (IHPC), Agency for Science, Technology and Research (A*STAR), 1 Fusionopolis Way, $\#$16-16 Connexis, Singapore 138632, Republic of Singapore}
\newcommand{\qinc}{
Quantum Innovation Centre (Q.InC), Agency for Science Technology and Research (A*STAR), 2 Fusionopolis Way, Innovis $\#$08-03, Singapore 138634, Republic of Singapore }

\newcommand{\VT}{Department of Computer Science, Virginia Polytechnic Institute and State University, Blacksburg, Virginia 24061, USA.\looseness=-1}
\newcommand{\ISI}{Cryptology and Security Research Unit, Indian Statistical Institute, Kolkata 700108, India.}

\begin{document}

\normalem
\newlength\figHeight 
\newlength\figWidth 

\title{Approximate Dynamical Quantum Error-Correcting Codes}

\author{Nirupam Basak}
\email{nirupambasak2020@iitkalumni.org}
\affiliation{\ISI}

\author{Andrew Tanggara}
\email{andrew.tanggara@gmail.com}
\affiliation{\CQT}
\affiliation{\NTU}

\author{Ankith Mohan}
\email{ankithmo@vt.edu}
\affiliation{\VT}
\affiliation{\ihpc}

\author{Goutam Paul}
\email{goutam.paul@isical.ac.in}
\affiliation{\ISI}

\author{Kishor Bharti}
\email{kishor.bharti1@gmail.com}
\affiliation{\qinc}
\affiliation{\ihpc}

\begin{abstract}
Quantum error correction plays a critical role in enabling fault-tolerant quantum computing by protecting fragile quantum information from noise. While general-purpose quantum error correction codes are designed to address a wide range of noise types, they often require substantial resources, making them impractical for near-term quantum devices. Approximate quantum error correction provides an alternative by tailoring codes to specific noise environments, reducing resource demands while still maintaining noise-robustness. Dynamical codes, including Floquet codes, introduce a dynamic approach to quantum error correction, employing time-dependent operations to stabilize logical qubits. In this work, we combine the flexibility of dynamical codes with the versatility of approximate quantum error correction to offer a promising avenue for addressing dominant noise in quantum systems. We construct several approximate dynamical codes using the recently developed strategic code framework. As a special case, we recover the approximate static codes widely studied in the existing literature. By analyzing these approximate dynamical codes through semidefinite programming, we establish the uniqueness and robustness of the optimal encoding, decoding, and check measurements. We also develop a temporal Petz recovery map suited to approximate dynamical codes. 
\end{abstract}

\maketitle

\section{Introduction}
Quantum computers hold the promise of solving classically intractable problems~\cite{shor1999polynomial,arora2009computational,nielsen2001quantum}. However, running large algorithms on quantum processors requires extremely low error rates, often as low as $10^{-10}$~\cite{kivlichan2020improved,campbell2021early}. Current quantum devices fall significantly short of this benchmark. Quantum error correction (QEC) offers a pathway~\cite{shor1995scheme,gottesman1997stabilizer,dennis2002topological,kitaev2003fault,lidar2013quantum,terhal2015quantum}. to overcome these limitations by encoding the information of a logical qubit across multiple physical qubits using entanglement. 

While QEC provides a framework for fault-tolerant quantum computing, it comes with substantial resource costs. Implementing QEC requires many additional qubits and complex operations, making it resource-intensive. General-purpose QEC codes are designed to handle a broad range of noise types, but this universality often makes them inefficient for specific noise environments, which are common in practical systems. Approximate QEC~\cite{bennett1996mixed,leung1997approximate,fletcher2007optimum,schumacher2002approximate,ng2010simple, faist2020continuous, hayden2020approximate, zhou2020optimal, beny2010general, cafaro2014approximate, crepeau2005approximate, yi2024complexity, klesse2007approximate, sang2024approximate, ben2011approximate, renes2016uncertainty, zhao2024extracting, buscemi2008entanglement, mandayam2012towards, jayashankar2018pretty, biswas2024noise, dutta2024noise} addresses this limitation by focusing on the dominant noise in the system. 
They achieve protection at levels that are similar to general-purpose codes while using resources more efficiently, thus are better suited for near-term quantum devices with limited resources.

Recently, a new class of codes, called dynamical codes~\cite{hastings2021dynamically,davydova2023floquet,haah2022boundaries,gidney2021fault,gidney2022benchmarking,hilaire2024enhanced,vuillot2021planar,paetznick2023performance,fu2024error,fahimniya2023hyperbolic,higgott2023constructions,tanggara2024simple,kesselring2024anyon,bombin2023unifying,davydova2023quantum,aasen2023measurement,dua2024engineering,zhang2023x,berthusen2023partial,ellison2023floquet,sullivan2023floquet}, has been introduced. These codes take a dynamic approach to quantum error correction. The first example, the Floquet quantum code~\cite{hastings2021dynamically}, uses a periodic schedule of non-commuting two-qubit measurements. This creates a time-evolving codespace that stabilizes logical qubits. Dynamical codes have several advantages. They use simpler low-weight parity checks, achieve higher error thresholds~\cite{gidney2021fault,gidney2022benchmarking,hilaire2024enhanced}, and allow finite encoding rates in specific variants, like hyperbolic Floquet codes~\cite{fahimniya2023hyperbolic,higgott2023constructions,tanggara2024simple}. These features make them highly promising for resource-efficient and fault-tolerant quantum computing.

As discussed earlier, approximate QEC has shown that tailoring codes to specific noise models can effectively reduce resource requirements without losing its robustness to errors. This raises an important question: can the efficiency of approximate QEC be integrated into the framework of dynamical codes? Approximate dynamical codes could offer a powerful solution for addressing dominant noise sources while retaining the flexibility and resource efficiency of dynamical codes. Despite their potential, the design of approximate dynamical codes remains an open problem.

In this work, we present constructions for various approximate dynamical codes. These codes are designed using the recently introduced strategic code framework~\cite{tanggara2024strategic}, a universal spatio-temporal approach to quantum error correction. Leveraging the properties of semidefinite programs, we prove that the optimal encoding, decoding, and check measurements for approximate codes—both static and dynamical—are unique and robust. Additionally, we develop a temporal Petz recovery map specifically tailored for approximate dynamical codes. As a special case, we recover the approximate static codes widely studied in the existing literature.

\subsection{Relation to prior work}

We note that our results complement some prior work on quantum error-correction.
Firstly in Theorem~\ref{thm:primal_unique}, we show the uniqueness and robustness of the optimal solution to the optimization program proposed in~\cite{tanggara2024strategic}, which we use to construct the approximate dynamical codes in Section~\ref{sec:approx_code_ampl_damp}.
Secondly we show an information-theoretic sufficient condition on whether approximate dynamical error-correction can be done  (Theorem~\ref{thm:info_theoretic_AQECC}), analogous to the conditions given in~\cite{tanggara2024strategic}.
Third, we propose a definition of temporal Petz-recovery map in Section~\ref{sec:temporal_petz} which we show to be the optimal decoder for correctable errors (Theorem~\ref{thm:petz_rec}), generalizing the results for Petz recovery map~\cite{barnum2002reversing, dutta2024noise, ng2010simple}.

\medskip

\section{Universal framework for quantum error-correction}
Quantum error-correcting code has been conventionally defined by an \textit{encoding} isometry $V$ that maps a quantum state $|\psi\>$ to \textit{codestate} $|\Bar{\psi}\>:= V|\psi\>$ which belongs to a \textit{codespace} $Q$, a subspace of a Hilbert space $A$ describing the system where information is encoded in.
A codespace $Q$ correct errors arising from an assumed noise model represented by a \textit{noise channel} $\mathcal{E}$ whenever there exists a \textit{decoding channel} $\mathcal{D}$ that ``reverses'' the noise as $\mathcal{D}\circ\mathcal{E}(V\ketbra{\psi}V^\dag) = \ketbra{\psi}$. 
A dynamical code, on the other hand, consist of a sequence of multiple \textit{rounds} $r\in\{0,1,2,\dots\}$ with codespaces $Q_0,Q_1,Q_2,\dots$, 
where the evolution $Q_{r-1} \mapsto Q_r$ is determined by a quantum operation $\mathcal{C}^{(r)}$ that map the codestates in $Q_{r-1}$ to codestates in $Q_r$, whereas $Q_0$ is the initial codespace defined by encoder $\mathcal{C}^{(0)}$.
Operations $\mathcal{C}^{(1)},\mathcal{C}^{(2)},\dots$ often correspond to measurements that extract error syndromes as it evolves the codespace as in dynamical codes, although in principle might also consist of gates, ancilla system preparations, and removal of subsystems.
In this case, analysis of correctable noise occurring over multiple time points is more complicated as it may exhibit \textit{temporal} (non-Markovian) correlations, on top of spatial correlations. 

\begin{figure}
    \centering
    \includegraphics[width=0.9\columnwidth]{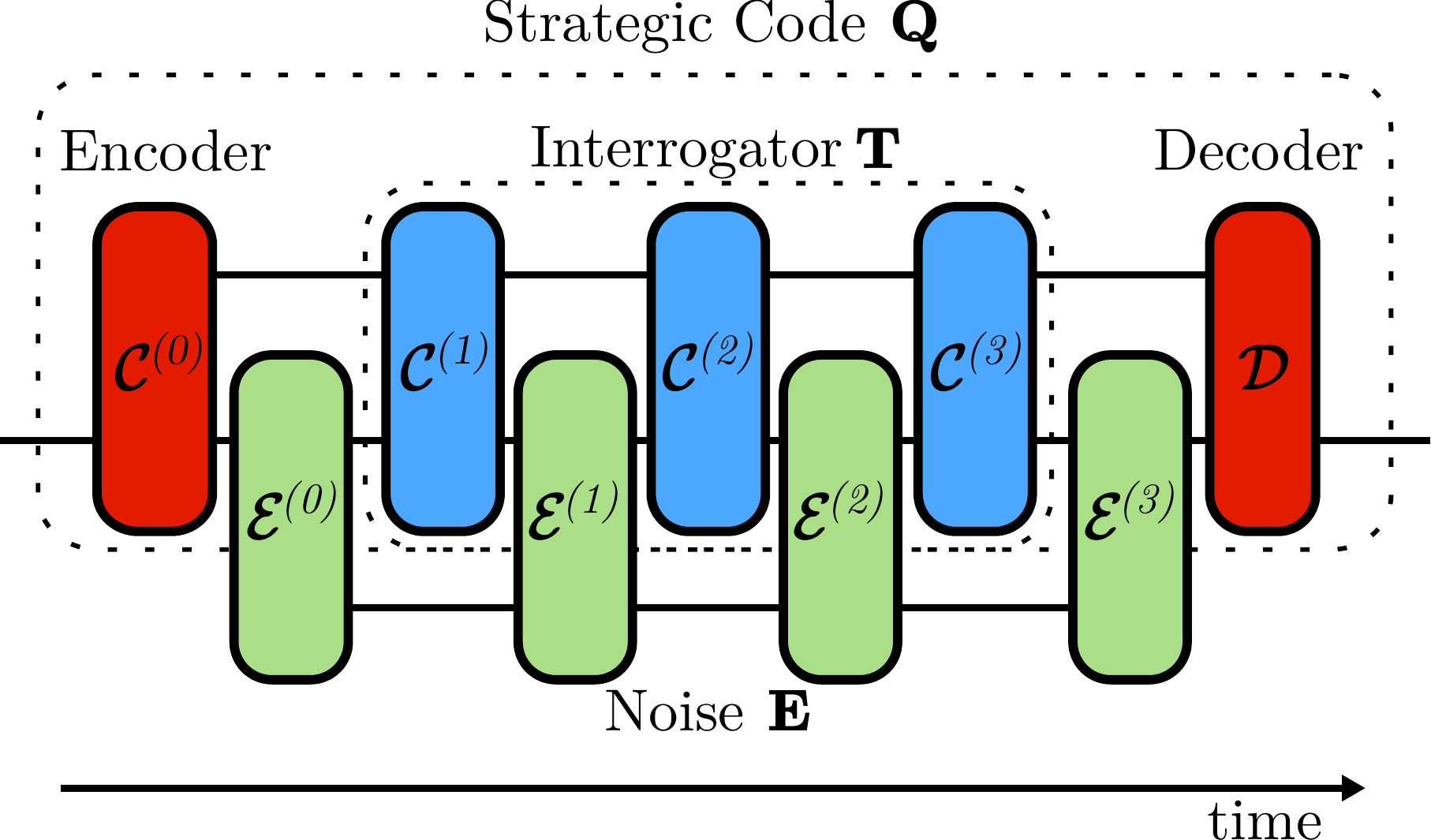}
    \caption{Strategic code $\mathbf{Q}$ with encoder $\mathcal{C}^{(0)}$, decoder $\mathcal{D}$, and three rounds of operations $\mathcal{C}^{(1)},\mathcal{C}^{(2)},\mathcal{C}^{(3)}$ performed by interrogator $\mathbf{T}$ and correlated noise $\mathbf{E}$ with errors $\mathcal{E}^{(0)},\dots,\mathcal{E}^{(3)}$.
    The top wire represents memory maintained by the strategic code which accounts for any adaptive operations performed by the interrogator and correction operations performed by the decoder.
    The middle wire represents the sequence of systems where quantum information is encoded in.
    The bottom wire represents the noise environment that accounts for any temporal correlation between errors $\mathcal{E}^{(r)}$ and $\mathcal{E}^{(r')}$.}
    \label{fig:strategic_code_general}
\end{figure}

The recently proposed universal framework for quantum error-correction known as the \textit{strategic code} framework~\cite{tanggara2024strategic} allows for a convenient representation of dynamical codes and corresponding temporally-correlated (non-Markovian~\cite{prakash2024characterizing, hakoshima2021relationship, oreshkov2007continuous, pollock2018non, pollock2018operational, milz2021quantum}) noise. 
The strategic code framework revolves around the object called the \textit{interrogator} $\mathbf{T}$ which represents any sequence of $l$ rounds of operations $\mathcal{C}^{(1)},\dots,\mathcal{C}^{(l)}$ performed in a quantum error-correction (QEC) procedure between encoding $\mathcal{C}^{(0)}$ and decoding $\mathcal{D}$.
All three objects $(\mathcal{C}^{(0)},\mathbf{T},\mathcal{D})$ defines a \textit{strategic code} $\mathbf{Q}$.
The framework also allows for adaptive operations, namely where round-$r$ operation $\mathcal{C}^{(r)}$ can be \textit{adaptively} conditioned on measurement outcomes $m=m_1,\dots,m_{r-1}$ in rounds $1,\dots,r-1$.
Interrogator $\mathbf{T}$ thus admits a decomposition $\mathbf{T}=\sum_m\mathbf{T}_m$ where $\mathbf{T}_m := |C_m\rr\ll C_m|$ is a rank-1 operator representing a sequence of operations corresponding to $m$. 
On the other hand, a sequence of (possibly correlated) noise $\mathcal{E}^{(0)},\dots,\mathcal{E}^{(l)}$ can described by multi-time noise $\mathbf{E}$. 
An example of a strategic code with a three-round interrogator is illustrated in Fig.~\ref{fig:strategic_code_general}.
We note that a static code in this framework can be simply captured by removing the interrogator $\mathbf{T}$, leaving us with a strategic code $\mathbf{Q}$ with encoder $\mathcal{C}^{(0)}$ and decoder $\mathcal{D}$ (as described in Appendix~\ref{app:static_QECC}).
For details on the strategic code and how the objects are represented formally, see Appendix~\ref{app:strategic_code_details}. 

\medskip

\section{Approximate strategic code optimization}
By using the quantum combs formalism~\cite{chiribella2009theoretical,gutoski2007toward,chiribella2008quantum,milz2021quantum,oreshkov2012quantum}, we can represent both noise $\mathbf{E}$ and strategic code $\mathbf{Q}$ as linear operators which allow for a bi-linear operation ``$\ast$'' between them resulting in a quantum channel $\mathbf{E}\ast\mathbf{Q}$.
Analogous to how correctable noise $\mathcal{E}$ can be reversed in static codes, i.e. $\mathcal{D}\circ\mathcal{E}(V\ketbra{\psi}V^\dag)=\ketbra{\psi}$, we say that a noise $\mathbf{E}$ is correctable by a strategic code $\mathbf{Q}$ whenever $\mathbf{E}\ast\mathbf{Q}(\ketbra{\psi})=\ketbra{\psi}$.
A relaxation of these correctability conditions where we only require the output state $\rho' := \mathbf{E}\ast\mathbf{Q}(\rho)$ to be ``close'' to input state $\rho$ gives us an \textit{approximate} strategic code $\mathbf{Q}$, analogous to approximate static codes~\cite{ng2010simple,fletcher2007optimum} (where $\rho' := \mathcal{D}\circ\mathcal{E}(V\rho V^\dag)$).
As shown in~\cite{tanggara2024strategic}, the quantum combs formulation allow us to obtain an approximate strategic code $\mathbf{Q}$ for a given noise $\mathbf{E}$ by means of optimization over linear operators that gives a valid strategic code.

Here we focus on strategic codes with an interrogator $\mathbf{T}$ performing a single operation $\mathcal{C}^{(1)}$ under noise $\mathbf{E}$ with error operations $\mathcal{E}^{(0)}$ after the encoding round and $\mathcal{E}^{(1)}$ between operation $\mathcal{C}^{(1)}$ and decoding round.
We refer to this strategic code as a \textit{single-round} strategic code.
The ensemble of logical quantum states is described by density operator $\rho$ over a logical Hilbert space $L$.
The single-round strategic code consists of an encoding channel $\mathcal{C}^{(0)}:\mathcal{L}(L)\rightarrow\mathcal{L}(A_0)$ and an operation $\mathcal{C}^{(1)}$.
Operation $\mathcal{C}^{(1)}$ is a quantum instrument $\left\{ \mathcal{C}_m^{(1)} \right\}_m$ where $\mathcal{C}_m^{(1)}:\mathcal{L}(A_0)\rightarrow\mathcal{L}(A_1)$ is a completely-positive map such that $\sum_m \mathcal{C}_m^{(1)}$ is trace-preserving.
Here we denote $\mathcal{L}(A)$ as the space of bounded linear operators that map states in Hilbert space $A$ to itself.
Also, the decoding channel is of the form $\mathcal{D}_m:\mathcal{L}(A_1)\rightarrow\mathcal{L}(L)$ where $L'$ indicates the system which the decoded quantum state is in.
Error maps are of the form $\mathcal{E}^{(0)}:\mathcal{L}(A_0)\rightarrow\mathcal{L}(A_0\otimes E_0)$ and $\mathcal{E}^{(1)}:\mathcal{L}(A_0\otimes E_0)\rightarrow\mathcal{L}(A_1)$ where $E_0$ denotes the noise environment that accounts for any temporal correlations.
The index $m$ in $\left\{ \mathcal{C}_m^{(1)} \right\}_m$ and $\mathcal{D}_m$ indicates information extracted by $\mathcal{C}^{(1)}$ (through some measurement) which is propagated to the decoding stage to determine which decoding channel $\left\{\mathcal{D}_m\right\}_m$ should be used.

By using the Choi operator representation $\mathbf{C}^{(0)},\left\{ \mathbf{C}_m^{(1)} \right\}_m,\{\mathbf{D}_m\}_m$ for the encoding, round-1 operation, and decoding, respectively, we can express the strategic code as
\begin{equation}\label{eqn:single-round_code}
\begin{aligned}
    \mathbf{Q} = \sum_m \mathbf{C}^{(0)} \otimes \mathbf{C}_m^{(1)} \otimes \mathbf{D}_m \;.
\end{aligned}
\end{equation}
Operator $\mathbf{Q}$ is positive semidefinite as a consequence of each $\mathbf{D}_m,\mathbf{C}_m^{(1)} , \mathbf{C}^{(0)}$ being positive semidefinite. 
On the other hand, noise $\mathbf{E}$ is also described by a positive semidefinite operator that admits a generic form of $\mathbf{E} = \sum_e |E_e\rr\ll E_e|$.
Each $E_e$ is a linear operator of the form $A_0\otimes A_1 \mapsto A_0\otimes A_1$ that maps the input code systems of $\mathcal{E}_0,\mathcal{E}_1$ to their output code systems, and $|E_e\rr := \sum_{j_0,j_1} E_e|j_0,j_1\>|j_0,j_1\>$ is the vectorized form of operator $E_e$ and $j_0,j_1$ are basis of $A_0,A_1$, respectively.

Now to formulate our optimization problem, consider a maximally-mixed state $\rho$ on initial logical system $L$ and the channel induced by $\mathbf{E}\ast\mathbf{Q}$.
We use the \emph{entanglement fidelity} $F_{ent}(\rho,\mathbf{E}\ast\mathbf{Q})$ as our performance metric, which quantifies how close the output state $\mathbf{E}\ast\mathbf{Q}(\rho)$ is to the input state $\rho$. 
This entanglement fidelity is given by
\begin{equation}\label{eqn:entanglement_fidelity}
\begin{aligned}
    F_{ent}(\rho,\mathbf{E}\ast\mathbf{Q})  
    &= \tr\big( \mathbf{E}\ast\mathbf{Q} \, |\rho\rr\ll\rho| \big) \;, 
\end{aligned}
\end{equation}
where $|\rho\rr := \sum_j \rho|j\>|j\>$ is the vectorized form of density operator $\rho$.
Essentially, the optimization consists of maximizing the entanglement fidelity over all valid strategic code, which takes the form of eqn.~\eqref{eqn:single-round_code}.
Its validity can be formulated as a set of \textit{positivity} and \textit{normalization} conditions.
Namely, positivity requires that $\mathbf{C}^{(0)},\mathbf{C}_m^{(1)},\mathbf{D}_m$ are all positive semidefinite, whereas normalization requires that $\tr_{A_1}\left( \sum_m\mathbf{C}_m^{(1)} \right) = I_{A_0}$, $\tr_{A_0}\left( \mathbf{C}^{(0)} \right) = I_L$, and $\tr_{L}(\mathbf{D}_m) = I_{A_1}$, which also disallow signalling backwards in time.
Both of these conditions altogether force the strategic code to consist only of a sequence of physical operations.

Finally, the complete optimization problem to obtain an approximate strategic code for a given noise $\mathbf{E}$ and initial state $\rho$ is given by
\begin{equation}\label{eqn:max_entanglement_fidelity1}
\begin{gathered}
    \max_{\mathbf{Q}} \tr\big( \mathbf{E}\ast\mathbf{Q} \, |\rho\rr\ll\rho| \big) \\ 
    \textup{such that} \\
    \mathbf{Q} = \sum_{m} \mathbf{D}_m \otimes \mathbf{C}_m^{(1)} \otimes \mathbf{C}^{(0)} \\
    \mathbf{D}_m \geq 0 \;,\quad \tr_{L}(\mathbf{D}_m) = I_{A_1} \\
    \mathbf{C}^{(0)} \geq0 \;,\quad \tr_{A_0}\big(\mathbf{C}^{(0)}\big) = I_L \\
    \mathbf{C}_m^{(1)} \geq0 \;,\quad \tr_{A_1}\big(\sum_m \mathbf{C}_m^{(1)}\big) = I_{A_0} \;.
\end{gathered}
\end{equation} 
A common approach to solve for the optimal Choi operators $\Dmat_{m}^*,\ \Cmat^{(0) *},\ \Cmat_m^{(1) *}$ in program~\eqref{eqn:max_entanglement_fidelity1}, is through the \emph{see-saw algorithm}.
First, we initialize two of the Choi operators, say $\Dmat_m$ and $\Cmat^{(0)}$, and solve program~\eqref{eqn:max_entanglement_fidelity1} to obtain the solution $\hat{\Cmat}_m^{(1)}$.
Next, we use this initialization $\Cmat^{(0)}$, set $\Cmat_m^{(1)} = \hat{\Cmat}_m^{(1)}$, and solve the program to get $\hat{\Dmat}_m$.
Then, setting $\Cmat_m^{(1)} = \hat{\Cmat}_m^{(1)}$ and $\Dmat_m = \hat{\Dmat}_m$, we solve the program to obtain $\hat{\Cmat}^{(0)}$.
We repeat these steps until we fail to see any significant progress.
This is described more formally in Algorithm~\ref{algo:see-saw} in Appendix~\ref{app:primal_unique}.

Theorem~\ref{thm:primal_unique} shows that the Choi operators obtained through this process are unique and robust.
The proof of this theorem is provided in Appendix~\ref{app:primal_unique}.
Appendix~\ref{app:gen_QECC}
extends this to the general case consisting of multiple rounds of check operations.

\begin{theorem}\label{thm:primal_unique}
    The optimal Choi operators $\Cmat^{(0) *}, \Cmat_m^{(1) *}, \Dmat_{m}^*$ obtained through solving program~\eqref{eqn:max_entanglement_fidelity1}, corresponding to the encoding channel, check instrument and the decoding channel respectively, are unique and robust, i.e,
    \begin{equation*}
        \begin{aligned}
            \|\tilde{\Cmat}^{(0)} - \Cmat^{(0) *}\|_F &\leq \mathcal{O}(\epsilon), \\
            \|\tilde{\Cmat}_m^{(1)} - \Cmat_m^{(1) *}\|_F &\leq \mathcal{O}(\epsilon), \quad \text{and} \\
            \|\tilde{\Dmat}_m - \Dmat_m^*\|_F &\leq \mathcal{O}(\epsilon),
        \end{aligned}
    \end{equation*}
    for any $\tilde{\Cmat}^{(0)},\ \tilde{\Cmat}_m^{(1) *}$, and $\tilde{\Dmat}_m$ that are feasible in program~\eqref{eqn:max_entanglement_fidelity1}. 
    Here $\|A\|_F = \sqrt{\tr(A^\dag A)}$ denotes the Frobenius norm.
\end{theorem}

\medskip

\section{Amplitude-damping noise}\label{sec:approx_code_ampl_damp}
We consider $\mathbf E$ as an amplitude damping noise with damping strength $\gamma$. Since an amplitude damping channel is temporally self-similar~\cite{utagi2020temporal}, the noise maps $\mathcal E^{(l)}$ for the $l$-th round have similar forms as $\mathbf E$. The explicit form of the operators is provided in Appendix~\ref{app:AD_noise}.

We consider the local noise model where the noise randomly impacts $k$ out of $n$ physical qubits. Then, the corresponding noise model can be written as
\begin{equation}
\label{eq:two_round_AD_gen}
    \mathbf{E} = \sum_{K\in\mathcal Q_k} \sum_{i,j}\bigotimes_{k'\in K} \frac{1}{|\mathcal Q_k|}
    |E_{i,k'}^{(0)}\rr\ll E_{i,k'}^{(0)}| \otimes |E_{j,k'}^{(1)}\rr\ll E_{j,k'}^{(1)}|.
\end{equation}
where $\mathcal Q_k$ contains all possible combinations of choosing $k$ out of $n$ physical qubits.
Here, $E_{i,k'}^{(0)}$ and $E_{j,k'}^{(1)}$ are linear operators which maps $A_0$ and $A_1$ to itself, respectively. 

Figures~\ref{fig:fidelity_plot_1} and~\ref{fig:fidelity_plot_2} show plots of the entanglement fidelity $F_{ent}$ of a single logical qubit as a function of the damping strength $\gamma$, evaluated using Algorithm~\ref{algo:see-saw_k} for $2-, 3-$, and $4$-qubit codes with $k$-qubit noise (eqn.~\eqref{eq:two_round_AD_gen}). 
Here, we consider only a single intermediary check round between the encoder and the decoder without any memory effect, that is, $m=1$ in program~\eqref{eqn:max_entanglement_fidelity1}. 
We observe that our approximate dynamical codes outperform certain approximate static codes, specifically the $\llbracket5, 1\rrbracket$ and $\llbracket3, 1\rrbracket$ codes from Ref.~\cite{dutta2024smallest}, even when these encoders are paired with the Petz recovery map decoder~\cite{barnum2002reversing, dutta2024noise, ng2010simple} to correct single-qubit errors.
The $\llbracket4, 1\rrbracket$ code from Ref.~\cite{leung1997approximate} performs better than our approximate dynamical code. 
A possible reason for such behavior might be that, in our case, the program~\eqref{eqn:max_entanglement_fidelity1} corresponding to the approximate dynamical code, is reporting a local maxima. 

We also consider amplitude-damping noise with specific weights~\cite{cafaro2014approximate,jayashankar2022achieving,dutta2024smallest}. The key distinction between this model and the previously discussed local noise model lies in how the noise is applied. In the local noise model, single-qubit amplitude-damping noise is independently applied to any $k$ qubits, while the remaining $n - k$ qubits undergo identity operations. In contrast, in the weight-$k$ amplitude-damping model, the \emph{no-error operator} $E^{(l)}_{0,*}$ is applied to exactly $k$ qubits, and the \emph{error operator} is applied to the remaining $n - k$ qubits. Note that both models coincide when $k = n$, in which case the noise is referred to as \emph{full amplitude-damping noise} (see Appendix~\ref{app:AD_noise} for detailed discussion). Corresponding plots are shown in Appendix~\ref{app:low_weight}.
Moreover, for the all-qubit noise scenario, we observe no improvement in fidelity when compared with the case where no encoding operation was performed. 
Although in this case, the codes obtained from program~\eqref{eqn:max_entanglement_fidelity1} provide better fidelity than the existing $\llbracket5, 1\rrbracket$ and $\llbracket3, 1\rrbracket$ codes from Ref.~\cite{dutta2024smallest}, and also for the cases where these encoders are paired with the Petz recovery map decoder. 
The corresponding plot is illustrated in Appendix~\ref{app:all_error}.

We notice that adding memory does not improve the performance of the codes for the noise model being considered here.
The encoder, decoder and check instrument for each of these cases are provided in Section~\ref{sec:algo_QEC}.
These findings suggest that the presented codes offer a better fidelity in the presence of amplitude-damping noise in comparison with the previously studied static codes. 

\begin{figure}
\centering
\includegraphics[width=\columnwidth]{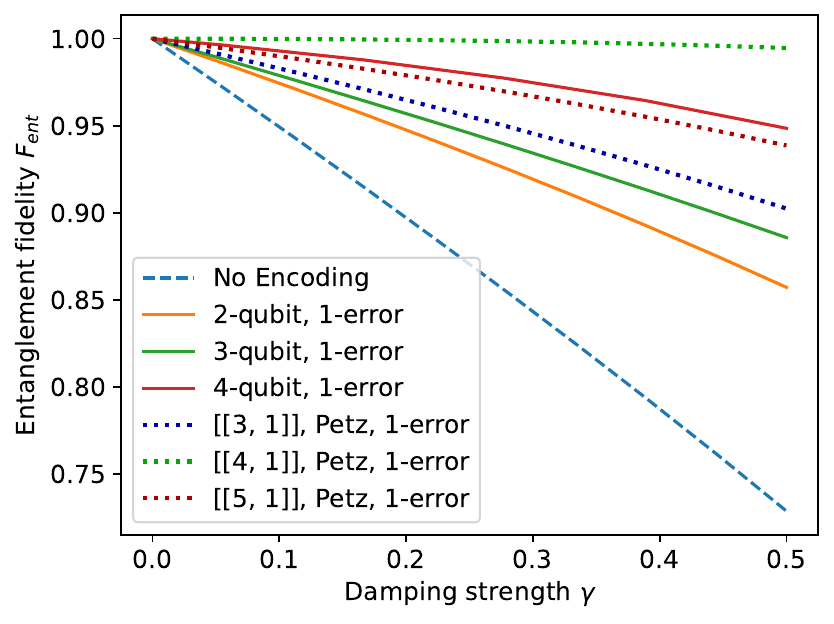}
\caption{\label{fig:fidelity_plot_1}
The entanglement fidelity $F_{ent}$ versus the damping strength $\gamma$ is analyzed using Algorithm~\ref{algo:see-saw_k} for $2$-, $3$-, and $4$-qubit codes with a single check operator, i.e., $m=1$, in the single-qubit noise scenario where anyone from $n$ physical qubits is erroneous.
Observe that these codes outperform the $\llbracket5, 1\rrbracket$ and $\llbracket3, 1\rrbracket$ codes from Ref.~\cite{dutta2024smallest}, along with the cases where these encoders are paired with the Petz recovery map decoders, in the presence of single-qubit errors. 
The encoders, decoders and check operators obtained for each of these cases are described in Section~\ref{sec:algo_QEC}. 
}
\end{figure}

\begin{figure}
\centering
\includegraphics[width=\columnwidth]{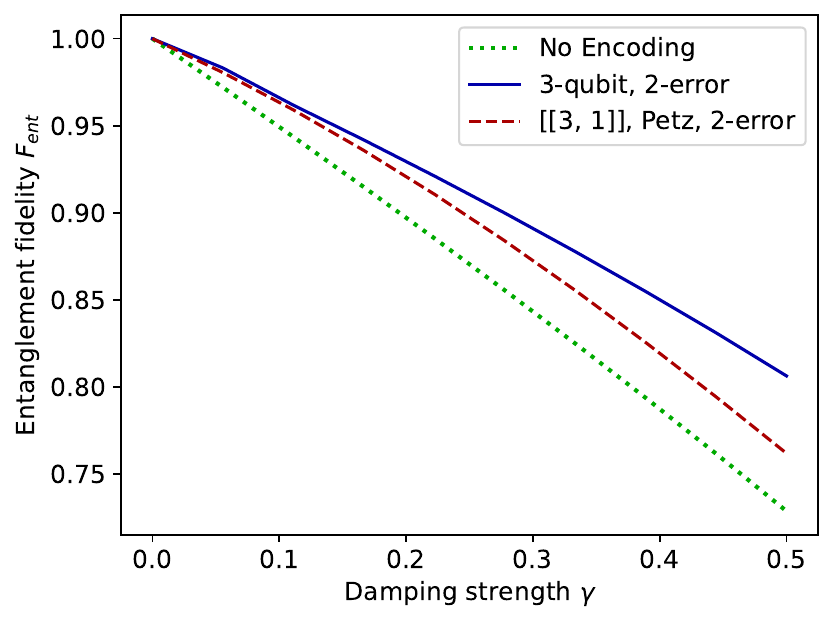}
\caption{\label{fig:fidelity_plot_2}
The entanglement fidelity $F_{ent}$ versus the damping strength $\gamma$ is plotted using Algorithm~\ref{algo:see-saw_k} for $3$-qubit code with a single check operator, i.e., $m=1$, in the $2$-qubit noise scenario where any $2$ from $n$ physical qubits are erroneous.
Observe that these codes outperform the $\llbracket3, 1\rrbracket$ code from Ref.~\cite{dutta2024smallest}, along with the cases where these encoders are paired with the Petz recovery map decoders, in the presence of 2-qubit errors. 
The encoders, decoders and check operators obtained for each of these cases are described in Section~\ref{sec:algo_QEC}.
}
\end{figure} 

\medskip

\section{Temporal Petz recovery map}\label{sec:temporal_petz}

As in~\cite{barnum2002reversing, dutta2024noise, ng2010simple}, \textit{Petz recovery map decoding} for a static code with codespace $Q$ under noise channel $\mathcal{E}(\cdot) = \sum_e E_e (\cdot) E_e^\dag$ is defined as quantum channel $\mathcal{R}_{Q,\mathcal{E}}$ with Kraus operators 
\begin{equation}\label{eqn:petz_recovery_static_code}
    R_{Q,\mathcal{E},j} = \Pi_Q E_e^\dag \mathcal{E}(\Pi_Q)^{-\frac{1}{2}} \;,
\end{equation}
where $\Pi_Q$ is projector onto codespace $Q$.
It is shown in~\cite[Lemma 2]{ng2010simple} that whenever errors from $\mathcal{E}$ can be corrected by code $Q$, the Petz recovery map $\mathcal{R}_{Q,\mathcal{E}}$ described above is an optimal decoding map.
As argued in~\cite{ng2010simple}, the optimality of the Petz recovery map is closely related to the Knill-Laflamme necessary and sufficient condition~\cite{knill1997theory} for the perfect correctability of noise $\mathcal{E}$ over code $Q$.
Particularly, the Petz recovery map $\mathcal{R}_{Q,\mathcal{E}}$ is equivalent to the perfect recovery map $\mathcal{R}_\mathrm{perf}$ which one can construct whenever the Knill-Laflamme condition is satisfied for the pair $Q,\mathcal{E}$. 

For dynamical codes or more general codes defined over multiple timesteps, it would be valuable to obtain an analogue of the Petz recovery map, which can be used to analyze and benchmark an \textit{approximate} dynamical code $\mathbf{Q}$ with respect to noise $\mathbf{E}$.
Such analysis and benchmarking have been done for static codes in~\cite{ng2010simple}, but have been absent for dynamical codes.
One could try a similar approach for dynamical codes by starting from a necessary and sufficient condition for correctability of noise $\mathbf{E}$ over code $\mathbf{Q}$.
Fortunately, such a condition has been shown in the strategic code framework~\cite[Theorem 1]{tanggara2024strategic}, which generalizes the Knill-Laflamme condition to dynamical codes. 

\begin{fact}(\cite{tanggara2024strategic}, Theorem 1)\label{thm:algebraic_KL_condition}
    A strategic code $\mathbf{Q}$ with initial codespace $Q_0$ and interrogator $\mathbf{T}=\sum_m|C_m\rr\ll C_m|$ corrects noise $\mathbf{E}=\sum_e|E_e\rr\ll E_e|$ if and only if
    \begin{equation}\label{eqn:algebraic_KL_condition}
    \begin{gathered}
        \ll E_{e'}| \big( |C_m\rr\ll C_m| \otimes |\Bar{j}\>\<\Bar{i}| \big) |E_e\rr = \lambda_{e',e,m} \delta_{j,i}
    \end{gathered}
    \end{equation}
    for a constant $\lambda_{e',e,m}\in\C$, for all $m$, all pairs of error sequences $e,e'$, and all $i,j$.
    Here, $|\Bar{i}\>,|\Bar{j}\>$ are a pair of basis states of initial codespace $Q_0$.
\end{fact} 

Now we propose a generalization of the Petz recovery map for static codes in eqn.~\eqref{eqn:petz_recovery_static_code} to the dynamical code scenario, which we call the \textit{temporal Petz recovery map}.
Later, we show that the perfect recovery channel which one can construct whenever the condition in Fact~\ref{thm:algebraic_KL_condition} is satisfied coincides with the temporal Petz recovery map as defined below, analogous to how static Petz recovery map coincides with the static perfect recovery map.

\begin{definition}\label{def:temporal_petz_recovery}
    Consider a strategic code $\mathbf{C}$ with an initial codespace $Q_0$ with projetor $\Pi_{Q_0}$ and interrogator $\mathbf{T}$ which admits decomposition $\mathbf{T} = \sum_m |C_m\rr\ll C_m|$, and noise $\mathbf{E}$ with quantum channel representation defined by $\mathbf{E}(A) = \sum_e \mathbf{E}_e A \mathbf{E}_e^\dag$ (see Appendix~\ref{app:strategic_code_details:purified}).
    A \textit{temporal Petz recovery} with respect to $(Q_0,\mathbf{T},\mathbf{E})$ is a family of quantum channels $\left\{ \mathcal{R}_{Q_0,\mathbf{T},\mathbf{E},m}^\mathrm{Petz} \right\}_m$, each defined by its Kraus operators $\{R_{e|m}\}_e$ where
    \begin{equation}\label{eqn:general_petz_recovery}
        R_{e|m} = (\ll C_m|\otimes\Pi_{Q_0}) F_e^\dag 
     \Big(\mathbf{E}\big( |C_m\rr\ll C_m| \otimes \Pi_{Q_0} \big)\Big)^{-\frac{1}{2}} \;.
    \end{equation}
\end{definition}

This natural generalization is obtained by viewing the multi-timestep noise $\mathbf{E}$ as a channel that maps the operator defined by the interrogator $\mathbf{T}$ and initial codespace $Q_0$ to the state at the start of the recovery stage (in system $A_l$).
We argue that this temporally generalized Petz recovery channel is well motivated for the following two reasons.
First, in the static code scenario where only initial codespace $Q_0$ and a noise channel $\mathcal{E}$ is involved (i.e. no interrogator is present), then we recover the Petz recovery map in eqn.~\eqref{eqn:petz_recovery_static_code}.
Namely, the operator $|C_m\rr\ll C_m| \otimes \Pi_{Q_0}$ and $\ll C_m|\otimes\Pi_{Q_0}$ in eqn.~\eqref{eqn:general_petz_recovery} are simply the initial codespace projector $\Pi_{Q_0}$ and $\{E_e\}_e$ are the Kraus operators of noise channel $\mathcal{E}$. 
Secondly, temporal Petz recovery $\left\{ \mathcal{R}_{Q_0,\mathbf{T},\mathbf{E},m}^\mathrm{Petz} \right\}_m$ is precisely the perfect recovery $\left\{ \mathcal{R}_{Q_0,\mathbf{T},\mathbf{E},m}^* \right\}_m$ for a strategic code $\mathbf{C}$ that corrects noise $\mathbf{E}$.
This is analogous to how the Petz recovery map coincides with the perfect recovery map in the static code scenario~\cite[Lemma 2]{ng2010simple}.

\begin{theorem}\label{thm:petz_rec}
    If strategic code $\mathbf{C}$ corrects noise $\mathbf{E}$, both admitting representations as in Definition~\ref{def:temporal_petz_recovery}, then the dynamical Petz recovery map $\mathcal{R}_{Q_0,\mathbf{T},\mathbf{E},m}^\mathrm{Petz}$ is equal to the perfect recovery map $\mathcal{R}_{Q_0,\mathbf{T},\mathbf{E},m}^*$, for each $m$.
\end{theorem}

The proof of Theorem~\ref{thm:petz_rec} is given in Appendix~\ref{app:proof_petz_recovery_optimal_for_correctable_noise}.

\medskip

\section{Approximate quantum error-correcting conditions}

Theorem $2$ of~\cite{tanggara2024strategic} gives an information-theoretic condition for perfect error correction by a strategic code $\Cmat$.
In the following, $S(\rho)$ denotes the von Neumann entropy of density operator $\rho$.

\begin{fact}\label{thm:info_theoretic_QECC}(\cite{tanggara2024strategic}, Theorem $2$)
    A strategic code $\mathbf{C}$ with initial codespace $Q_0$ and interrogator $\mathbf{T}$ corrects noise $\mathbf{E}$ if and only if
    \begin{equation}
        S(\rho_{m}^{R B_l E}) = S(\rho_{m}^{R}) + S(\rho_{m}^{B_l E}) \;.
    \end{equation} 
\end{fact}

This condition considers the purification of all operations performed within the strategic code $\mathbf{C}$, as well as noise $\mathbf{E}$ given one-half of maximally entangled state input $\ket{\phi}_{R Q_0} = \sum_i \ket{i}_{R} \ket{i}_{Q_0}$ on the initial codespace $Q_0$ and reference system $R$ with the same dimension as $A_0$.
Here we consider ancillary systems $B_l$ used by the interrogator $\mathbf{T}$ and the noise environment $E$ at the beginning of the decoding stage.
For a given $m$, then the global state over systems $R,B_l,A_l,E$ at the output of the dynamics induced by the interaction between interrogator $\mathbf{T}$ and noise $\mathbf{E}$ is a pure state $|\mathbf{E}\ast\mathbf{T}_m(\phi)\>$.
We denote the density operator of the reduced $|\mathbf{E}\ast\mathbf{T}_m(\phi)\>$ state on system $W$ as $\rho_m^W$ (where $W$ is a subset of systems $R,B_l,A_l,E$).
For details on this scenario, see Appendix~\ref{app:strategic_code_details:purified}.

Suppose that instead of $S(\rho_m^{R B_l E}) = S(\rho_m^{R}) + S(\rho_m^{B_l E})$ for all $m$, we have $S(\rho_m^{R}) + S(\rho_m^{B_l E}) - S(\rho_m^{R B_l E}) < \varepsilon$.
Although perfect error correction is not possible, can we still approximately correct it?
If so, with what fidelity can we restore the initial code state which is one-half of the maximally entangled state $\ket{\phi}_{R Q_0}$ by operation on $A_l$ for all $l$?
Theorem~\ref{thm:info_theoretic_AQECC} bounds the entanglement fidelity of this restoration.

\begin{theorem}\label{thm:info_theoretic_AQECC}
    Suppose that 
    $$S(\rho_m^{R}) + S(\rho_m^{B_l E}) - S(\rho_m^{R B_l E}) < \varepsilon$$ for all $m$. Then the entanglement fidelity between the recovered code-state $\sigma^{R A_0}$ and the initial code-state $\ket{\phi}_{R Q_0}$ is bounded as
    $$F_{ent}(\sigma^{R A_0}, \ket{\phi}_{R Q_0}) \geq 1 - 2\sqrt{\varepsilon}.$$
\end{theorem}
Appendix~\ref{app:info_theoretic_AQECC}
describes the proof of Theorem~\ref{thm:info_theoretic_AQECC}.

\medskip

\section{\label{sec:algo_QEC}See-saw algorithm and QEC Protocol}

Algorithm~\ref{algo:see-saw_k} describes the see-saw algorithm for single-check dynamical QECC optimization when we are dealing with amplitude-damping noise as defined in~\eqref{eq:two_round_AD_gen}.  

\begin{algorithm*}[htpb]
\hrulefill
\vskip-10pt
\caption{See-saw algorithm for single-check dynamical QECC optimization for $k$-qubit error}\label{algo:see-saw_k}
\vskip-6pt\hrulefill
    \begin{algorithmic}[1]
        \State \text{Input:} $\Cmat^{(0)},\ \{\Cmat_m^{(1)}\}_m$
        \State Compute $N_0 = \left[ \otimes_{k'\in K} E_{i,k'}^{(0)*} \ket{a}\bra{b} E_{i,k'}^{(0)\top},\ \forall\ a, b, i, K \right]$
        \State Compute $N_1 = \left[ \otimes_{k'\in K} E_{j,k'}^{(0)*} \ket{c}\bra{d} E_{j,k'}^{(0)\top},\ \forall\ c, d, j, K \right]$
        \State Compute $R=\left[\rho \ket{s}\bra{t} \rho^\dag,\ \forall s,t\right]$
        \Repeat
            \State Compute $\Cmat_0\Cmat_1 = \left[\sum\limits_{abcdijstK} \bra{tb} \Cmat^{(0)} \ket{sa}\ \tr \left[ \Cmat_m^{(1)} \left(\left(N_0\right)_{abiK} \otimes \ket{c}\bra{d}\right) \right]\left(\left(N_1\right)_{cdjK}\otimes R_{st}\right),\ \forall\ m \right]$
            \State\label{eq:opt_Dmat_m_1} Solve for optimal $\hat{\Dmat}_m: \max \left\{ \sum\limits_m \tr \left[ \Dmat_m \left(\Cmat_0\Cmat_1\right)_m\right]\ :\ \Dmat_m \geq 0,\ \tr_{L}(D_m) = I_{A_1} \right\}$
            \State Compute $\Cmat_0\Dmat = \left[\sum\limits_{abcdijstK} \bra{tb} \Cmat^{(0)} \ket{sa} \tr \left[ \hat{\Dmat}_m \left(\left(N_1\right)_{cdjK}\otimes R_{st}\right) \right]\left(\left(N_0\right)_{abiK} \otimes \ket{c}\bra{d}\right) ,\ \forall\ m \right]$
            \State\label{eq:opt_Cmat_m1_1} Solve for optimal $\hat{\Cmat}_m^{(1)}: \max \left\{ \sum\limits_{m} \tr \left[ \Cmat_m^{(1)} \left(\Cmat_0\Dmat\right)_m\right]\ :\ \Cmat_m^{(1)} \geq 0,\ \tr_{A_1} \left( \sum\limits_{m} \Cmat_m^{(1)} \right) = I_{A_0} \right\}$
            \State Compute $\Cmat_1\Dmat = \sum\limits_{abcdijmstK} \tr \left[ \hat{\Cmat}_m^{(1)} \left(\left(N_0\right)_{abiK} \otimes \ket{c}\bra{d}\right) \right]\tr \left[ \hat{\Dmat}_m \left(\left(N_1\right)_{cdjK}\otimes R_{st}\right) \right]|sa\rangle\langle tb|$
            \State\label{eq:opt_Cmat_m0_1} Solve for optimal $\hat{\Cmat}^{(0)}: \max \left\{ \tr \left[\Cmat^{(0)}\Cmat_1\Dmat\right]\ :\ \Cmat^{(0)} \geq 0,\ \tr_{A_0}(\Cmat^{(0)}) = I_L \right\}$
        \Until{\label{abc}no progress is being made}
        \State \textbf{Return} $\{\hat{\Dmat}_m\}_m, \{\hat{\Cmat}_m^{(1)}\}_m, \hat{\Cmat}^{(0)}$
    \end{algorithmic}
\end{algorithm*}

\medskip

\subsection{QEC protocol for amplitude-damping noise}

\begin{itemize}
    \item\textbf{2-qubit code:}
        \begin{itemize}
            \item\textbf{Encoder: }$\ket{0_L}=\frac{1}{\sqrt{2}}(\ket{00}+\ket{11})_{ma},\ket{1_L}=\frac{1}{\sqrt{2}}(\ket{00}-\ket{11})_{ma}$, where $m$ denotes the main qubit and $a$ denotes the auxiliary qubit.
            \item\textbf{Check instrument: }$\ket{00}_{ma}\to\frac{1}{\sqrt{2}}(\ket{00}+\ket{10})_{ma},\ket{01}_{ma}\to\ket{01}_{ma},\ket{10}_{ma}\to\frac{1}{\sqrt{2}}(\ket{00}-\ket{10})_{ma},\ket{11}_{ma}\to\frac{1}{\sqrt{2}}(\ket{00}-\ket{10})_{ma}$.
            \item\textbf{Decoder: } $\ket{01}_{ma}\to\frac{1}{\sqrt{2}}(\ket{00}-\ket{10})_{ma},\ket{11}_{ma}\to\frac{1}{\sqrt{2}}(\ket{00}-\ket{10})_{ma}$. Then the auxiliary qubit would be traced out.
        \end{itemize}

    \item\textbf{3-qubit code:}
        \begin{itemize}
            \item\textbf{Encoder: }$\ket{0_L}=\frac{1}{\sqrt{2}}(\ket{000}+\ket{111})_{ma_1a_2},\ket{1_L}=\frac{1}{\sqrt{2}}(\ket{000}+\ket{111})_{ma_1a_2}$, where $m$ denotes the main qubit and $a_i$ denote the auxiliary qubits.
            \item\textbf{Check instrument: }$\ket{101}_{ma_1a_2}\to\frac{1}{2}(\ket{100}+\ket{101})_{ma_1a_2},\ket{110}_{ma_1a_2}\to\frac{1}{2}(\ket{100}+\ket{110})_{ma_1a_2},\ket{111}_{ma_1a_2}\to\ket{100}_{ma_1a_2}$. Other state would remain unchanged.
            \item\textbf{Decoder: }$\ket{000}_{ma_1a_2}\to\frac{1}{\sqrt{2}}(\ket{000}+\ket{100})_{ma_1a_2}$. All other states would go to $\frac{1}{\sqrt{2}}(\ket{000}-\ket{100})_{ma_1a_2}$. After this operation, the auxiliary qubits would be traced out.
        \end{itemize} 
\end{itemize}

\medskip

\paragraph*{Software.---}

    The Python code for Algorithm~\ref{algo:see-saw_k} and the Choi representations corresponding to the encoder, check operation and decoder, are available in the following GitHub repository: \href{https://github.com/sp-k/QEC_approximate_dynamical.git}{https://github.com/sp-k/QEC\_approximate\_dynamical.git}. 

\medskip

\section{Discussion} In this work, we introduced approximate dynamical codes, a new class of quantum error-correcting codes that merge the adaptability of approximate quantum error correction with the flexibility of dynamical codes. Using the strategic code framework, we established the uniqueness and robustness of optimal encoding, decoding, and check measurements via semidefinite programming. As a special case, we recover the approximate static codes~\cite{fletcher2007optimum} widely studied in the existing literature (see Appendix~\ref{app:static_QECC}). 

An interesting parallel exists between dynamical codes and dynamical decoupling (DD)~\cite{ezzell2023dynamical, viola1999dynamical, biercuk2009optimized, khodjasteh2005fault, de2010universal, du2009preserving, alvarez2011measuring, viola2003robust, medford2012scaling, khodjasteh2007performance, facchi2004unification}. While DD mitigates noise by applying unitary pulses that average out interactions, dynamical codes leverage time-evolving codespaces and strategic measurements to counteract noise structurally. Extending our optimization framework to uncover new DD sequences is an exciting direction.

While our work focuses on qubit-based codes, extending these to qudits could unlock new advantages. Incorporating quantum memory may enable adaptive error correction, particularly in non-Markovian environments where errors are temporally correlated. Scaling approximate dynamical codes to larger system sizes remains another interesting open challenge. While our methods demonstrate feasibility for small-scale systems, their applicability to large-scale quantum processors will require optimizing both computational efficiency and physical implementability.

A recent study established a rigorous connection between quantum circuit complexity and approximate quantum error correction capability~\cite{yi2024complexity}. It introduced a code parameter, referred to as subsystem variance, which is closely linked to the optimal precision of AQEC. These insights may provide a new perspective on approximate dynamical quantum error correction, potentially leading to further advancements in the field. 

\medskip

\paragraph*{Note added.---} While this work was being completed, a related study on dynamical codes for biased noise was posted on arXiv~\cite{setiawan2024tailoring}. Unlike their approach, our work leverages an optimization-theoretic framework, enabling the construction of smaller codes that are applicable to arbitrarily chosen noise models and support any number of check measurement rounds. 

\medskip
\section*{Acknowledgements} 
KB thanks the support from Q.InC Strategic Research and Translational Thrust. This work was performed when AM was an intern at IHPC, A*STAR.
AM was also supported by the Commonwealth of Virginia's Commonwealth Cyber Initiative (CCI) under grant number $469024$.
AT is supported by CQT PhD scholarship and the Google PhD Fellowship. We thank Hoang Xuan for interesting discussions.


% \begingroup
% \renewcommand{\addcontentsline}[3]{}
% \bibliographystyle{unsrt}
% \bibliography{references}
% \endgroup

\appendix
\section{Static QECC Optimization}\label{app:static_QECC}

\begin{figure}[htpb]
    \centering
    \tikzset{
        my label/.append style={below,yshift=-0.8cm}
    }
    \begin{tikzpicture}
    \node[scale=1]{
    \begin{quantikz}
        & & \gate[wires=2, nwires={1}, style={fill=blue!10}]{\mathcal{C}^{(0)}} \gategroup[2, steps=1, style={dashed, rounded corners, fill=blue!10, inner xsep=2pt}, background]{{Encoding $\mathbf{C}^{(0)}$}} & & \gate[wires=2, nwires={1}, style={fill=red!10}]{\mathcal{D}} \gategroup[2, steps=1, style={dashed, rounded corners, fill=red!10, inner xsep=2pt}, background]{{Decoding $\mathbf{D}$}} \gategroup[2, steps=1, style={dashed, rounded corners, inner xsep=2pt}, background]{} & \\
        & \gate[wires=1, nwires={1}, style={white}]{\rho} & & \gate[wires=2,  nwires={2}, style={fill=green!20, rounded corners}]{\mathcal{E}^{(0)}} & & \qw & \gate[style={white}]{\rho} \\
        & & \gate[wires=1, nwires={1}, style={white}]{|0\>} & \qw & & &
    \end{quantikz}
    };
    \end{tikzpicture}
    \caption{Static QECC illustration.
    A static code is defined by an encoder $\mathcal{C}^{(0)}:\mathcal{L}(L)\rightarrow\mathcal{L}(A_0)$ and decoder $\mathcal{D}:\mathcal{L}(A_0)\rightarrow\mathcal{L}(L)$ under noise $\mathcal{E}^{(0)}$ applied to the system $A_0$ at the output of $\mathcal{C}^{(0)}$ where quantum state $\rho\in\mathcal{L}(L)$ is encoded.
    Decoder $\mathcal{D}$ is applied to the noisy system to recover the initial state $\rho$ encoded by $\mathcal{C}^{(0)}$.
    }
    \label{fig:static_QECC}
\end{figure}

Figure~\ref{fig:static_QECC} describes a density operator $\rho$ in $\mathscr{H}_L$, an encoding channel $\mathcal{C}^{(0)}: \mathscr{H}_L \to \mathscr{H}_{Q_0}$, an error map $\mathcal{E}^{(0)}: \mathscr{H}_{Q_0} \to \mathscr{H}_{Q_0}'$, and a decoding channel $\mathcal{D}: \mathscr{H}_{Q_0}' \to \mathscr{H}_L'$.
The corresponding optimization problem is

\begin{equation}
    \begin{aligned}
        \max &\quad \tr \left[ \left( \mathbf{E}^T \otimes I_{\{L', L\}}a \right)\ \mathbf{Q}(| \rho \rangle \rangle \langle \langle \rho | \otimes I_{Q\backslash\{L, L'\}}) \right] \\
        \text{subject to} &\quad \mathbf{Q} = \mathbf{D} \otimes \mathbf{C}^{(0)} \\
        &\quad \mathbf{D} \geq 0, \qquad \tr_{L'} (\mathbf{D}) = I_{Q'_0} \\
        &\quad \mathbf{C}^{(0)} \geq 0, \qquad \tr_{Q_0} (\mathbf{C}^{(0)}) = I_{L}.
    \end{aligned}
\end{equation}
where $Q = \{L, Q_0, Q_0', L'\}$.

As shown in Appendix~\ref{app:obj}, the objective function can be written as
\begin{multline} 
    \sum_{ijkl} \bra{lj} \mathbf{C}^{(0)} \ket{ki}\ \tr \Bigg[ \mathbf{D}_{Q_0', L'}\Bigg( \left( \sum_{e_0} E_{e_0}^* \ket{i}\bra{j} E_{e_0}^\top \right)_{Q_0'}\\ \otimes \left( \rho \ket{k}\bra{l} \rho^\dag \right)_{L'} \Bigg) \Bigg]
\end{multline}

The corresponding see-saw algorithm is described in Algorithm~\ref{algo:static_see-saw}.

\begin{algorithm*}[htpb]
\hrulefill
\vskip-10pt
\caption{Static QECC see-saw}\label{algo:static_see-saw}
\vskip-6pt\hrulefill
    \begin{algorithmic}[1]
        \State \textbf{Input:} $\mathbf{C}^{(0)}$
        \State Compute $(N_0)_{ij} = \sum_{e_0} E_{e_0}^* \ket{i}\bra{j} E_{e_0}^\top\ \forall\ i, j$
        \State Compute $R_{kl} = \rho \ket{k}\bra{l} \rho^\dag\ \forall\ k, l$
        \Repeat
            \State Compute $\mathbf{\tilde{C}}_0 = \sum\limits_{ijkl} \bra{lj} \mathbf{C}^{(0)} \ket{ki}\ \left( (N_0)_{ij} \otimes R_{kl} \right)$
            \State Solve for optimal $\hat{\Dmat}: \max \left\{ \tr \left[\mathbf{D}\ \tilde{\mathbf{C}}_0 \right]:\ \mathbf{D} \geq 0, \quad \tr_{L'}(\mathbf{D}) = I_{Q_0'} \right\}$
            \State Compute $\mathbf{\tilde{D}} = \sum_{ijkl}\ket{ki}\bra{lj}\tr \left[ \mathbf{D}\ \left( (N_0)_{ij} \otimes R_{kl} \right) \right]$
            \State Solve for optimal $\hat{\Cmat}^{(0)}: \max \left\{ \tr\left[\mathbf{C}^{(0)}\mathbf{\tilde{D}}\right]:\ \mathbf{C}^{(0)} \geq 0, \quad \tr_{Q_0}(\mathbf{C}^{(0)}) = I_L \right\}$
        \Until{no progress is being made}
        \State \textbf{Return} $\hat{\Dmat}, \hat{\Cmat}^{(0)}$
    \end{algorithmic}
\end{algorithm*}

Figure~\ref{fig:Rep_13.1} compares the $\llbracket5, 1, 3\rrbracket$ code from Ref.~\cite{bennett1996mixed} and the $\llbracket4, 1\rrbracket$ approximate code from Ref.~\cite{leung1997approximate} using optimal recovery from Ref.~\cite{fletcher2007optimum} with their counterparts obtained by optimizing over the recovery using Algorithm~\ref{algo:static_see-saw}.

\begin{figure*}
\centering
\includegraphics[width=0.9\textwidth]{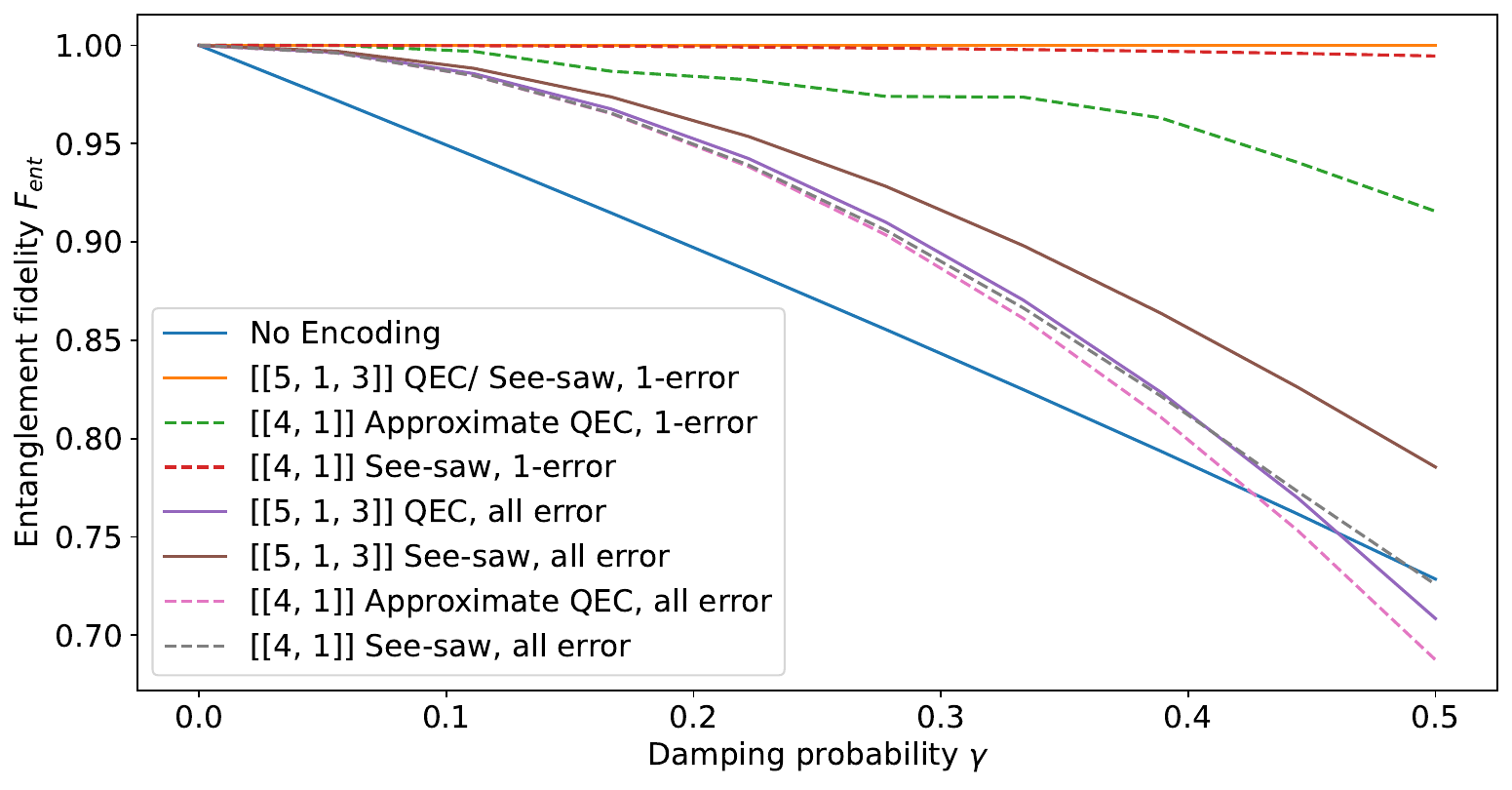}
\caption{damping strength vs Entanglement Fidelity plot using Algorithm~\ref{algo:static_see-saw} compared to two static codes $\llbracket5, 1, 3\rrbracket$ code from Ref.~\cite{bennett1996mixed} and $\llbracket4, 1\rrbracket$ approximate code from Ref.~\cite{leung1997approximate} with optimal recovery from Ref.~\cite{fletcher2007optimum}.}
\label{fig:Rep_13.1}
\end{figure*}

\section{Strategic Code Framework}\label{app:strategic_code_details} 

In this section, we describe the strategic code framework in detail.
Objects in the framework, the entire strategic code $\mathbf{C}$, interrogator $\mathbf{T}$, and noise $\mathbf{E}$ has a unique positive semidefinite representation which generalizes Choi representation for quantum channels (see~\cite{chiribella2008quantum,chiribella2009theoretical,milz2021quantum,oreshkov2012quantum}).
We detail how each of these objects are defined.

A strategic code $\mathbf{C}$ can be described by an encoder $\mathcal{C}^{(0)}$, an interrogator $\mathbf{T}$ consisting of a sequence of quantum operations $\mathcal{C}^{(1)},\dots,\mathcal{C}^{(l)}$, and a decoder $\mathcal{D}$.
Let $\mathcal{L}(A)$ denote the set of all bounded linear operator mapping Hilbert space $A$ to itself.
Encoder $\mathcal{C}^{(0)}$ is a quantum channel (completely positive and trace-preserving map) of the form $\mathcal{L}(L)\mapsto\mathcal{L}(A_0)$.
The first operation $\mathcal{C}^{(1)}$ is a quantum channel as well with Kraus operators $\{C_{m_1}^{(1)}\}_{m_1}$ (obtained from spectral decomposition of its Choi operator $\mathbf{C}^{(1)} = \sum_{m_1} |C_{m_1}^{(1)}\rr\ll C_{m_1}^{(1)}|$), where $m_1$ is to be interpreted as outcome of measurements. 
Adaptive operations may be performed by interrogator $\mathbf{T}$, hence operation $\mathcal{C}^{(2)}$ at round $2$ may be conditioned on outcome $m_1$.
Thus operation $\mathcal{C}^{(2)}$ consists of quantum channels $\{\mathcal{C}_{m_1}^{(2)}\}_{m_1}$, each conditioned on outcome $m_1$ from round $1$.
Similar to the channel in round 1, channel $\mathcal{C}_{m_1}^{(2)}$ has Kraus operators $\{C_{m_2|m_1}^{(2)}\}_{m_2}$ and a Choi operator $\mathbf{C}^{(1)} = \sum_{m_2} |C_{m_2|m_1}^{(1)}\rr\ll C_{m_2|m_1}^{(1)}|$.
In rounds $r>2$, operation $\mathcal{C}^{(r)}$ is defined similarly with dependency on outcomes in previous rounds $m_{1:r}=m_1,\dots,m_{r-1}$. 
Interrogator $\mathbf{T}$ can also be represented as a positive semidefinite operator which takes the form of
\begin{equation}\label{eqn:app:interrogator_spectral_decomp}
    \mathbf{T} = \sum_m \mathbf{T}_m = \sum_m |C_m\rr\ll C_m|
\end{equation}
where $m:=m_1,\dots,m_l$ and $|C_m\rr := \bigotimes_{r=1}^l |C_{m_r|m_{1:r}}^{(r)}\rr$.
Decoder $\mathbf{D}$ also consists of decoding channels $\{\mathcal{D}_m\}_m$, each with Choi operator representation $\mathbf{D}_m$.
Finally, the entire strategic code can be described by
\begin{equation}
    \mathbf{C} = \sum_m \mathbf{C}^{(0)} \otimes \mathbf{T}_m \otimes \mathbf{D}_m \;,
\end{equation}
which is a positive semidefinite operator since each of its tensor factors is a positive semidefinite operator. 

Similarly, noise $\mathbf{E}$ admits a positive semidefinite operator representation which can be obtained by a sequence of quantum channels $\mathcal{E}^{(0)},\dots,\mathcal{E}^{(l)}$ where $\mathcal{E}^{(r)}:\mathcal{L}(A_r\otimes E_r)\rightarrow\mathcal{L}(A_r\otimes E_{r+1})$ for $1\leq r<l$ and $\mathcal{E}^{(0)}:\mathcal{L}(A_0)\rightarrow\mathcal{L}(A_0\otimes E_{1})$ and $\mathcal{E}^{(l)}:\mathcal{L}(A_l\otimes E_l)\rightarrow\mathcal{L}(A_l)$.
Hilbert space $E_r$ represents noise environment in which temporal correlation from round $r-1$ to round $r$ occurs.
Positive semidefinite operator $\mathbf{E}$ then can be obtained by inputting one-half of (unnormalized) maximally entangled state $|\phi^{(r)}\> = \sum_{j=1}^{d_r}|j\>_{A_r}|j\>_{A_r}$ to the $A_r$ input of each $\mathcal{E}^{(r)}$ (where $d_r=\dim A_r$).
The collection of the outputs from $\mathcal{E}^{(0)},\dots,\mathcal{E}^{(l)}$ along with the other half of the maximally entangled states gives us a positive-semidefinite operator $\mathbf{E}\in\mathcal{L}(A_0^{\otimes2}\otimes\dots\otimes A_l^{\otimes2})$ (since each $\mathcal{E}^{(r)}$ is completely-positive).
We can take the spectral decomposition of $\mathbf{E}$ as
\begin{equation}
    \mathbf{E} = \sum_e |E_e\rr\ll E_e|
\end{equation}
where $|E_e\rr = \sum_{j} \alpha_{j} |j_0\>\otimes|j_0'\>\otimes\dots\otimes|j_l\>\otimes|j_l'\>$ with $\{|j_r\>\}_{j_r=1}^{d_r}$ a basis of input space $A_r$, $\{|j_r'\>\}_{j_r'=1}^{d_r}$ a basis for output space $A_r$ and $j=j_0,j_0',\dots,j_l,j_l'$.
Noise $\mathbf{E} = \sum_e |E_e\rr\ll E_e|$ can also be thought of as a quantum channel with Kraus operators $\mathbf{E}_e = \sum_{j} \alpha_{j} |j_l'\> (\<j_0|\otimes\<j_0'|\otimes\dots\otimes\<j_l|)$ which gives a channel that maps $\rho\otimes\mathbf{C}^{(0)} \otimes \mathbf{T}_m \mapsto \sum_e \mathbf{E}_e (\rho\otimes\mathbf{C}^{(0)}\otimes\mathbf{T}_m) \mathbf{E}_e^\dag$.

As a final note, the Choi operator $\mathbf{T}$ corresponding to quantum channel $\mathcal{T}:\mathcal{L}(A)\rightarrow\mathcal{L}(B)$ must also satisfy $\tr_B(\mathbf{T})=I_A$ due to its trace-preserving condition.
Thus this apply to each of the quantum channels in the interrogator $\mathbf{T}$, as well as encoding and decoding channels.
Therefore this condition appears as a constraint in the optimization problem in the main text to ensure the trace-preserving condition of these channels (along with the positive semidefiniteness condition).

\section{Proof of Theorem~\ref{thm:primal_unique}}
\label{app:primal_unique}

Consider the SDP
\begin{equation*}
    \mathrm{maximize:}\; \left\{ \tr(CX)\ :\ \Xi(X) = B, X \in \mathrm{Pos}(\mathcal{X}) \right\},
\end{equation*}
where $\Xi$ is any Hermiticity-preserving linear map, $B$ and $C$ are Hermitian matrices, and $\mathrm{Pos}(\mathcal{X})$ denotes the set of positive semidefinite matrices in a complex Euclidean space $\mathcal{X}$.
The dual of this SDP is expressed as
\begin{equation*}
    \begin{aligned}
        \mathrm{minimize:} \quad & \tr(BY) \\
        \mathrm{subject\ to} \quad & \Xi^*(Y) - C = Z \\
        & \hspace{10pt} Y \in \mathrm{Herm}(\mathcal{Y}) \\
        & \hspace{10pt} Z \in \mathrm{Herm}(\mathcal{Z}),
    \end{aligned}
\end{equation*}
for Hermitian matrices $Y$ and $Z$ in complex Euclidean spaces $\mathcal{Y}$ and $\mathcal{Z}$ respectively.
Here, $\mathrm{Herm}(\mathcal{Y})$ denotes the set of all Hermitian matrices in $\mathcal{Y}$.
Let $Z^*$ denote a dual optimal solution.

\begin{definition}[Dual nondegeneracy]
    \label{def:dual_nondeg}
     The dual solution $Z^*$ is said to be \emph{dual nondegenerate} if the homogeneous linear system
    \begin{equation}
        \label{eq:dual_non_deg}
            M Z^* = 0, \qquad \Xi(M) = 0
    \end{equation}
    only admits the trivial solution $M = 0$, for a symmetric matrix $M$.
\end{definition}

Ref.~\cite{alizadeh1997complementarity} showed that if $Z^*$ is a dual optimal and non-degenerate solution of a SDP, then there exists a \textit{unique} primal optimal solution.

\begin{fact}(\cite{bharti2019robust}, Theorem 6)\label{thm:robust}
    For a primal-dual pair of SDPs whose optimal values are equal and both obtained, assume that the set of primal feasible solutions is contained in a compact subset of the set of positive semidefinite operators.
    Denoting the set of primal operators as $\mathcal{P}$ with singularity degree $d$, we have that the distance of any primal feasible solution $\tilde{X}$ satisfying $p^* - \epsilon \leq \tr(CX)$ with $p^*$ denoting the primal optimal value, to the set $\mathcal{P}$, $\mathrm{dist}(\tilde{X}, \mathcal{P})$ satisfies
    \begin{equation}
        \mathrm{dist}(\tilde{X}, \mathcal{P}) \leq \mathcal{O}\left( \epsilon^{2^{-d}} \right).
    \end{equation}
\end{fact}

\bigskip

Algorithm~\ref{algo:see-saw} describes the see-saw algorithm for single-check dynamical QECC optimization. 

Theorem~\ref{thm:primal_unique}
 states that the Choi representation of the optimal decoding channel $\Dmat_m^*$ (line~\ref{eq:opt_Dmat_m} of Algorithm~\ref{algo:see-saw}), the Choi representation of the optimal check operation $\Cmat_m^{(1) *}$ (line~\ref{eq:opt_Cmat_m_1} of Algorithm~\ref{algo:see-saw}), and the Choi representation of the optimal encoding channel $\Cmat^{(0) *}$ (line~\ref{eq:opt_Cmat_m_0} of Algorithm~\ref{algo:see-saw}) is unique and robust.

Our approach is as follows.
We first derive the conditions that the optimal solution $Y_{D_m}^*$ of the dual program corresponding to the SDP in line~\ref{eq:opt_Dmat_m} of Algorithm~\ref{algo:see-saw} must satisfy.
Showing that $Y_{D_m}^*$ is non-degenerate proves that the corresponding optimal solution $\Dmat_m^*$ obtained by solving the SDP in line~\ref{eq:opt_Dmat_m} is unique.
As a consequence of Fact~\ref{thm:robust}, for any $\tilde{\Dmat}_m$ which is feasible in SDP in line~\ref{eq:opt_Dmat_m} of Algorithm~\ref{algo:see-saw} we get that
\begin{equation*}
    \|\tilde{\Dmat}_m - \Dmat_m^*\|_F = \mathrm{dist}\left( \tilde{\Dmat}_m, \Dmat_m^* \right) \leq \mathcal{O}(\epsilon),
\end{equation*}
thereby showing that the optimal solution $\Dmat_m^*$ is \textit{robust}.

By following the same argument, we can show the uniqueness and robustness of the optimal solutions $\Cmat_m^{(1) *}$ and $\Cmat^{(0) *}$ that result from solving the SDPs in line~\ref{eq:opt_Cmat_m_1} and line~\ref{eq:opt_Cmat_m_0} of Algorithm~\ref{algo:see-saw}.

\bigskip

Consider the primal SDP in line~\ref{eq:opt_Dmat_m} of Algorithm~\ref{algo:see-saw}, and its dual SDP,
\begin{widetext}
\allowdisplaybreaks
\begin{equation*}
    \begin{aligned}
        \mathrm{maximize:} &\; \left\{ \sum\limits_{m} \tr \left[ \Dmat_m  (\Cmat_0\Cmat_1)_m\right]\ :\ \Dmat_m \geq 0,\ \tr_{L'}(D_m) = I_{Q_1'} \right\} \\
        \mathrm{minimize:} &\; \left\{ \tr(Y_{D_m})\ :\ I_{L'} \otimes Y_{D_m} - (\Cmat_0 \Cmat_1)_m \geq 0 \right\}
    \end{aligned}
\end{equation*}
\end{widetext}
Suppose $\Dmat_m^*$ and $Y_{D_m}^*$ are optimal solutions for the primal and dual problems, respectively.

\begin{lemma}\label{lem:strong_dual}
    Strong duality holds for this primal-dual SDP pairs, i.e., $\sum_m \tr[\Dmat_m^* (\Cmat_0\Cmat_1)_m] = \tr(Y^*_{D_m})$.
\end{lemma}
\begin{proof}
    Let $\Dmat_m = I_{L'} \otimes I_{Q_1'}$. 
    This is positive definite and satisfies the constraint $\tr_{L'}(D_m) = I_{Q_1'}$.
    Thus, the primal SDP is strictly feasible.
    
    Let $Y_{D_m} = I_{Q_1'}$.
    Then we have $I_{L'} \otimes I_{Q_1'} - (\Cmat_0\Cmat_1)_m \geq 0$ which satisfies the constraint of the dual problem.
    So the dual SDP admits a feasible solution.
\end{proof}

\begin{lemma}\label{lem:Y_D_m}
In line~\ref{eq:opt_Dmat_m} of Algorithm~\ref{algo:see-saw}, $\Dmat_m^* = \sum_\alpha |d_{m \alpha}^* \rr\ll d_{m \alpha}^*|$ are Choi matrices of the optimal decoder if and only if
\begin{equation*}
    (I_{L'} \otimes Y_{D_m}^* - (\mathbf{C}_0 \mathbf{C}_1)_m) \mathbf{D}^*_m \geq 0\ \forall\ m
\end{equation*}
where
\begin{equation*}
    \begin{aligned}
        (\Cmat_0\Cmat_1)_m =& \sum\limits_{ijhkst} \bra{tj} \Cmat^{(0)} \ket{si}\ \tr \left[ \Cmat_m^{(1)} \left(N_0\right)_{ijhk} \right]\left(N_1\right)_{hkst}, \\
        (N_1)_{h k s t} =& \sum_\beta | (n_{1 \beta})_{hkst} \rr\ll n_{1 \beta})_{hkst} |, \\
        \bar{Y}_{D_m}^* \coloneqq& \sum_{\alpha \beta i j h k s t} \bra{tj} \Cmat^{(0)} \ket{si} \tr \left[ \Cmat_m^{(1)} (N_0)_{ijhk} \right]\\
            &\;\times (d_{m \alpha}^*)^\dag (n_{1 \beta})_{hkst} \tr \left[ (n_{1 \beta})_{hkst}^\dag d_{m \alpha}^*) \right ].
    \end{aligned}
\end{equation*}
\end{lemma}
\begin{proof}
    By complementary slackness, we have
    \begin{equation*}
        (I_{L'} \otimes Y_{D_m}^* - (\Cmat_0 \Cmat_1)_m) \Dmat_m^* = 0.
    \end{equation*}
    Let $\Dmat_m^*$ be expressed as $\Dmat_m^* = \sum_{\alpha} | d_{m \alpha}^* \rr\ll d_{m \alpha}^* |$.
    For all values of $\alpha$, we have
    \begin{equation*}
        \begin{aligned}
            (I_{L'} \otimes Y_{D_m}^*) | d_{m \alpha}^* \rr =& (\Cmat_0 \Cmat_1)_m | d_{m \alpha}^* \rr, \\
            | d_{m \alpha}^* \bar{Y}_{D_m}^* \rr =& \sum_{i j h k s t} \bra{tj} \Cmat^{(0)} \ket{si} \tr \left[ \Cmat_m^{(1)} (N_0)_{ijhk} \right]\\
            &\;\times(N_1)_{hkst} | d_{m \alpha}^* \rr.
        \end{aligned}
    \end{equation*}
    Let $N_1$ be written as $N_1 = \sum_\beta | n_{1 \beta} \rr \ll n_{1 \beta} |$. Then for all values of $\alpha$, we get
    \begin{equation*}
        \begin{aligned}
            | d_{m \alpha}^* \bar{Y}_{D_m}^* \rr =& \sum_{\beta i j h k s t} \bra{tj} \Cmat^{(0)} \ket{si} \tr \left[ \Cmat_m^{(1)} (N_0)_{ijhk} \right]\\
            &\;\times | (n_{1 \beta})_{hkst} \rr\ll (n_{1 \beta})_{hkst} | d_{m \alpha}^* \rr \\
            =& \sum_{\beta i j h k s t} \bra{tj} \Cmat^{(0)} \ket{si} \tr \left[ \Cmat_m^{(1)} (N_0)_{ijhk} \right]\\
            &\;\times | (n_{1 \beta})_{hkst} \rr \tr \left[ (n^\dag_{1 \beta})_{hkst} d_{m \alpha}^* \right] \\
            \implies d_{m \alpha}^* \bar{Y}_{D_m}^* =& \sum_{\beta i j h k s t} \bra{tj} \Cmat^{(0)} \ket{si} \tr \left[ \Cmat_m^{(1)} (N_0)_{ijhk} \right]\\
            &\;\times (n_{1 \beta})_{hkst} \tr \left[ (n_{1 \beta})^\dag_{hkst} d_{m \alpha}^* \right].
        \end{aligned}
    \end{equation*}
    Multiplying by $(d_{m \alpha}^*)^\dag$ on both sides and summing all the values of $\alpha$, we get
    \begin{equation*}
    \begin{aligned}
        \sum_\alpha (d_{m \alpha}^*)^\dag d_{m \alpha} \bar{Y}_{D_m}^* =& \sum_{\alpha \beta i j h k s t} \bra{tj} \Cmat^{(0)} \ket{si} \tr \left[ \Cmat_m^{(1)} (N_0)_{ijhk} \right]\\
            &\;\times (d_{m \alpha}^*)^\dag (n_{1 \beta})_{hkst} \tr \left[ (n_{1 \beta})^\dag_{hkst} d_{m \alpha}^* \right].
    \end{aligned}
    \end{equation*}
    Thus, we have
    \begin{equation}
    \label{eq:barY_D_m*}
    \begin{aligned}
        \bar{Y}_{D_m}^* =& \sum_{\alpha \beta i j h k s t} \bra{tj} \Cmat^{(0)} \ket{si} \tr \left[ \Cmat_m^{(1)} (N_0)_{ijhk} \right] (d_{m \alpha}^*)^\dag\\
            &\;\times (n_{1 \beta})_{hkst} \tr \left[ (n_{1 \beta})^\dag_{hkst} d_{m \alpha}^* \right].
    \end{aligned}
    \end{equation}
\end{proof}

\begin{algorithm*}[htpb]
\hrulefill
\vskip-10pt
\caption{See-saw algorithm for single-check dynamical QECC optimization}  
    \label{algo:see-saw}
    \vskip-6pt\hrulefill
    \begin{algorithmic}[1]
        \State \textbf{Input:} $\Cmat^{(0)},\ \{\Cmat_m^{(1)}\}_m$
        \State Compute $N_0 = \left[ \sum\limits_{e_0} E_{e_0}^* \ket{i}\bra{j} E_{e_0}^\top \otimes \ket{h}\bra{k},\ \forall\ i, j, h, k \right]$
        \State Compute $N_1 = \left[ \sum\limits_{e_1} E_{e_1}^* \ket{h}\bra{k} E_{e_1}^\top \otimes \rho \ket{s}\bra{t} \rho^\dag,\ \forall\ h, k, s, t \right]$
        \Repeat
            \State Compute $\Cmat_0\Cmat_1 = \left[\sum\limits_{ijhkst} \bra{tj} \Cmat^{(0)} \ket{si}\ \tr \left[ \Cmat_m^{(1)} \left(N_0\right)_{ijhk} \right]\left(N_1\right)_{hkst},\ \forall\ m \right]$
            \State\label{eq:opt_Dmat_m} Solve for optimal $\hat{\Dmat}_m: \max \left\{ \sum\limits_{m} \tr \left[ \Dmat_m  (\Cmat_0\Cmat_1)_m\right]\ :\ \Dmat_m \geq 0,\ \tr_{L'}(D_m) = I_{Q_1'} \right\}$
            \State Compute $\Cmat_0\Dmat = \left[\sum\limits_{ijhkst} \bra{tj} \Cmat^{(0)} \ket{si} \tr \left[ \hat{\Dmat}_m \left(N_1\right)_{hkst} \right]\left(N_0\right)_{ijhk},\ \forall\ m \right]$
            \State\label{eq:opt_Cmat_m_1} Solve for optimal $\hat{\Cmat}_m^{(1)}: \max \left\{ \sum\limits_{m}\ \tr \left[ \Cmat_m^{(1)} (\Cmat_0\Dmat)_m \right]\ :\ \Cmat_m^{(1)} \geq 0,\ \tr_{Q_1} \left( \sum\limits_{m} \Cmat_m^{(1)} \right) = I_{Q_0'} \right\}$
            \State Compute $\Cmat_1\Dmat = \sum\limits_{ijhkstm} \tr \left[ \hat{\Cmat}_m^{(1)} \left(N_0\right)_{ijhk} \right]\tr \left[ \hat{\Dmat}_m \left(N_1\right)_{hkst} \right]|si\rangle\langle tj|$
            \State\label{eq:opt_Cmat_m_0} Solve for optimal $\hat{\Cmat}^{(0)}: \max \left\{ \tr\left[ \Cmat^{(0)}\Cmat_1\Dmat\right]\ :\ \Cmat^{(0)} \geq 0,\ \tr_{Q_0}(\Cmat^{(0)}) = I_L \right\}$
        \Until{no progress is being made}
        \State \textbf{Return} $\{\hat{\Dmat}_m\}_m, \{\hat{\Cmat}_m^{(1)}\}_m, \hat{\Cmat}^{(0)}$
    \end{algorithmic}
\end{algorithm*}

\begin{lemma}\label{lem:Y_D_m_nondeg}
    The optimal dual solution
        \begin{equation*}
        \begin{aligned}
            \bar{Y}_{D_m}^* =& \sum_{\alpha \beta i j h k s t} \bra{tj} \Cmat^{(0)} \ket{si} \tr \left[ \Cmat_m^{(1)} (N_0)_{ijhk} \right] (d_{m \alpha}^*)^\dag\\
            &\;\times (n_{1 \beta})_{hkst} \tr \left[ (n_{1 \beta})^\dag_{hkst} d_{m \alpha}^*) \right ]
        \end{aligned}
        \end{equation*}
    is dual nondegenerate.
\end{lemma}
\begin{proof}
    From definition~\ref{def:dual_nondeg} it suffices to show that the homogeneous system $MY^* = 0$ and $\tr[M(I_{L'} \otimes Y^*)] = 0$ admits the trivial solution $M = 0$.

    For $MY^* = 0$, $M$ must lie in the null space of $Y^*$.
    However, $Y^*$ is constructed as a full-rank operator on the space spanned by the contributions of $(C_0C_1)_m$.
    Thus, having a trivial null space., implying that $M = 0$.

    For $\tr[M(I_{L'} \otimes Y^*)] = 0$ to hold, $M$ must be orthogonal to the subspace spanned by $I_{L'} \otimes Y$.
    Since $Y^*$ spans the operator space determined by $(C_0C_1)_m$, the extension $I_{L'} \otimes Y^*$ spans the corresponding space in the larger Hilbert space.
    The condition requires $M$ to be orthogonal to the entire operator space determined by $(C_0C_1)_m$, leaving us with the trivial solution $M = 0$.
\end{proof}

\begin{lemma}\label{lem:D_m_robust}
    The optimal primal solution $\Dmat_m^*$ is robust, i.e,
    \begin{equation}\label{eq:D_m_robust}
        \|\tilde{\Dmat}_m - \Dmat_m^*\|_F = \mathrm{dist}\left( \tilde{\Dmat}_m, \Dmat_m^* \right) \leq \mathcal{O}(\epsilon).
    \end{equation}    
\end{lemma}
\begin{proof}
    Note the following two observations.
    Firstly, the set of channels on a fixed finite-dimensional Hilbert space forms a compact set~\cite{watrous2018theory}.
    Secondly, Lemma~\ref{lem:strong_dual} shows the strict feasibility of the primal SDP, which corresponds to a singularity degree of zero.
    These observations in conjunction with Fact~\ref{thm:robust} result in Eq.~\eqref{eq:D_m_robust}.
\end{proof}

From Lemma~\ref{lem:Y_D_m} and Lemma~\ref{lem:Y_D_m_nondeg} we have that the optimal solution $\Dmat_m^*$ of the SDP in line~\ref{eq:opt_Dmat_m} of Algorithm~\ref{algo:see-saw} is unique, and Lemma~\ref{lem:D_m_robust} gives us its robustness.

Similarly, we can derive the conditions that the dual optimal solutions of the SDPs in line~\ref{eq:opt_Cmat_m_1} and line~\ref{eq:opt_Cmat_m_0} of Algorithm~\ref{algo:see-saw} must satisfy and show that these are dual non-degenerate.
Furthermore, we can show that the respective primal solutions $\Cmat_m^{(1) *}$ and $\Cmat_m^{(0)}$ are robust as well.
These are described in Lemma~\ref{lem:opt_Cmat_m_1} and Lemma~\ref{lem:opt_Cmat_m_0} respectively.
\begin{lemma}\label{lem:opt_Cmat_m_1}
    In line~\ref{eq:opt_Cmat_m_1} of Algorithm~\ref{algo:see-saw}, $\Cmat_m^{(1)*} = \sum_{\gamma_1} |c_{m \gamma_1}^{(1)*} \rr\ll c_{m \gamma_1}^{(1)*}|$ are Choi matrices of the optimal check operation if and only if
    \begin{equation*}
        (I_{Q_1} \otimes Y_{C_1}^* - (\Cmat_0 \Dmat)_m) \Cmat^{(1) *}_m \geq 0\ \forall\ m
    \end{equation*}
    where
    \begin{equation*}
        \begin{aligned}
            (\Cmat_0 \Dmat)_m =& \sum_{ijhkst} \bra{tj} \Cmat^{(0)} \ket{si}\ \tr[\Dmat_m (N_1)_{hkst}]\ (N_0)_{ijhk}, \\
            (N_0)_{i j h k} =& \sum_\delta | (n_{0 \delta})_{ijhk} \rr\ll n_{0 \delta})_{ijhk} |, \\
            \bar{Y}_{C_1}^* \coloneqq &\sum_{\gamma_1 \delta i j h k s t} \bra{tj} \Cmat^{(0)} \ket{si}\ \tr \left[ \Dmat_m (N_1)_{hkst} \right]\\
            &\;\times (c_{m \gamma_1}^{(1)*})^\dag (n_{0 \delta})_{ijhk} \tr \left[ (n_{0 \delta})^\dag_{ijhk} c_{m \gamma_1}^{(1)*} \right ].
        \end{aligned}
    \end{equation*}
    Moreover, $Y_{C_1}^*$ is dual non-degenerate and $\Cmat_m^{(1)*}$ is robust.
\end{lemma}
~
\begin{lemma}\label{lem:opt_Cmat_m_0}
    In line~\ref{eq:opt_Cmat_m_0} of Algorithm~\ref{algo:see-saw}, $\Cmat^{(0)*} = \sum_{\gamma_0} |c_{\gamma_0}^{(0)*} \rr\ll c_{\gamma_0}^{(0)*}|$ are Choi matrices of the optimal encoder if and only if
    \begin{equation*}
        (I_{Q_0} \otimes Y_{C_0}^* - \Cmat_1 \Dmat) \Cmat^{(0)*} \geq 0
    \end{equation*}
    where
    {\allowdisplaybreaks
    \begin{equation*}
        \begin{aligned}
            &\Cmat_1 \Dmat = \sum_{ijhkstm} \tr[\Cmat_m^{(1)} (N_0)_{ijhk}]\ \tr[\Dmat_m (N_1)_{hkst}] \ket{si}\bra{tj}, \\
            &\ket{si}\bra{tj} = \sum_f | f \rr\ll f |, \\
            &\bar{Y}_{C_0}^* \coloneqq \sum_{\gamma_0 f i j h k s t m} \tr \left[ \Cmat_m^{(1)} (N_0)_{ijhk} \right]\ \tr \left[ \Dmat_m (N_1)_{hkst} \right] \\
            &\qquad\qquad\times(c_{\gamma_0}^{(0)*})^\dag f\ \tr \left[ f^\dag c_{\gamma_0}^{(0)*} \right ].
        \end{aligned}
    \end{equation*}}
    Moreover, $Y_{C_0}^*$ is dual non-degenerate  and $\Cmat^{(0)*}$ is robust.
\end{lemma}

\section{General Dynamical QECC Optimization}~\label{app:gen_QECC}

Algorithm~\ref{algo:see-saw_gen} describes the see-saw algorithm for the general case comprising multiple rounds of check operations.
Extending the discussion from the previous section, we show the uniqueness of the optimal Choi operators $\Cmat_{m_0}^{(0)}, \{\Cmat_{m_1 | m_0}^{(1)}\}_{m_1}, \dots, \{\Cmat_{m_l | m_{l-1}}^{(l)}\}_{m_l}, \{\Dmat_{m_l}\}_{m_l}$ corresponding to the encoding channel, $l$ intermediate check operations and the decoding channel respectively.

\begin{algorithm*}[htpb]
\hrulefill
\vskip-10pt
\caption{See-saw algorithm for general dynamical QECC optimization}\label{algo:see-saw_gen}
\vskip-6pt\hrulefill
    \begin{algorithmic}[1]
        \State \textbf{Input:} $\Cmat^{(0)}_{m_0},\ \{\Cmat_{m_1|m_0}^{(1)}\}_{m_1},\dots,\ \{\Cmat_{m_l|m_{l-1}}^{(l)}\}_{m_l}$
        \State Compute $N_r = \left[ \sum\limits_{e_r} E_{e_r}^* \ket{i}\bra{j} E_{e_r}^\top \otimes \ket{h}\bra{k},\ \forall\ i, j, h, k \right]$ for $r=0,1,\dots,l-1$
        \State Compute $N_l = \left[ \sum\limits_{e_l} E_{e_l}^* \ket{h}\bra{k} E_{e_l}^\top \otimes \rho \ket{s}\bra{t} \rho^\dag,\ \forall\ h, k, s, t \right]$
        \Repeat
            \State Compute $\Cmat = \left[\sum\limits_{i_{0:l}j_{0:l}stm_{1:l-1}} \bra{tj_0} \Cmat^{(0)}_{m_0} \ket{si_0}\ \prod\limits_{r\neq l}\tr \left[ \Cmat_{m_{r+1}|m_r}^{(r+1)} \left(N_r\right)_{i_rj_ri_{r+1}j_{r+1}} \right]\left(N_l\right)_{i_lj_lst},\ \forall\ m_l \right]$
            \State\label{eq:opt_Dmat_m_l} Solve for optimal $\hat{\Dmat}_{m_l}: \max \left\{ \sum\limits_{m_l} \tr \left[ \Dmat_{m_l}  \Cmat_{m_l}\right]\ :\ \Dmat_{m_l} \geq 0,\ \tr_{L'}(D_{m_l}) = I_{Q_1'} \right\}$
            \For{$k=l-1,l-2,\dots,0$}
                \State Compute $\Cmat_k\Dmat = \left[\sum\limits_{i_{0:l}j_{0:l}stm_{1:k,k+2:l}} \bra{tj_0} \Cmat^{(0)}_{m_0} \ket{si_0}\ \prod\limits_{r\neq k,l}\tr \left[ \Cmat_{m_{r+1}|m_r}^{(r+1)} \left(N_r\right)_{i_rj_ri_{r+1}j_{r+1}} \right]\tr \left[ \hat{\Dmat}_{m_l} \left(N_l\right)_{i_lj_lst} \right]\left(N_k\right)_{i_kj_ki_lj_l} \right],$ for all $m_{k+1}$.
                \State\label{eq:opt_Cmat_m_k} Solve for optimal $\hat{\Cmat}_{m_{k+1}|m_k}^{(k+1)}: \max \left\{ \sum\limits_{m_{k+1}}\ \tr \left[ \Cmat_{m_{k+1}|m_k}^{(k+1)} (\Cmat_k\Dmat)_m \right]\ :\ \Cmat_{m_{k+1}|m_k}^{(k+1)} \geq 0,\ \tr_{Q_{k+1}} \left( \sum\limits_{m_{k+1}} \Cmat_{m_{k+1}|m_k}^{(k+1)} \right) = I_{Q_k'} \right\}$
            \EndFor
            \State Compute $\Cmat_l\Dmat = \sum\limits_{i_{0:l}j_{0:l}stm_{1:l}} \prod\limits_{r\neq l}\tr \left[ \hat{\Cmat}_{m_{r+1}|m_r}^{(r+1)} \left(N_r\right)_{i_rj_ri_{r+1}j_{r+1}} \right]\tr \left[ \hat{\Dmat}_{m_l} \left(N_l\right)_{i_lj_lst} \right]|si_0\rangle\langle tj_0|$
            \State\label{eq:opt_Cmat_m_0l} Solve for optimal $\hat{\Cmat}^{(0)}_{m_0}: \max \left\{ \tr\left[ \Cmat^{(0)}_{m_0}\Cmat_l\Dmat\right]\ :\ \Cmat^{(0)}_{m_0} \geq 0,\ \tr_{Q_0}(\Cmat^{(0)}_{m_0}) = I_L \right\}$
        \Until{no progress is being made}
        \State \textbf{Return} $\{\hat{\Dmat}_{m_l}\}_{m_l},\ \{\hat{\Cmat}_{m_1|m_0}^{(1)}\}_{m_1},\dots,\ \{\hat{\Cmat}_{m_l|m_{l-1}}^{(l)}\}_{m_l},\ \hat{\Cmat}^{(0)}_{m_0}$
    \end{algorithmic}
\end{algorithm*}

\begin{lemma}\label{lem:opt_Dmat_m_l}
    In line~\ref{eq:opt_Dmat_m_l} of Algorithm~\ref{algo:see-saw_gen}, $\Dmat_{m_l}^* = \sum_{\alpha} |d_{m_l \alpha}^* \rr\ll d_{m_l \alpha}^*|$ are Choi matrices of the optimal decoder if and only if
    \begin{equation*}
        (I_{L'} \otimes Y_{D_{m_l}}^* - \Cmat_{m_l}) \Dmat^*_{m_l} \geq 0\ \forall\ m_l
    \end{equation*}
    where
    {\allowdisplaybreaks
    \begin{equation*}
        \begin{aligned}
            \Cmat_{m_l} =& \sum_{i_{0:l} j_{0:l} s t m_{1:l-1}} \bra{tj_0} \Cmat_{m_0}^{(0)} \ket{si_0} \prod_{r \neq l} \tr[\Cmat_{m_{r+1} | m_r}^{(r+1)}\\
            &\;\times(N_r)_{i_r j_r i_{r+1} j_{r+1}}] (N_l)_{i_l j_l s t}, \\
            \left(N_l\right)_{i_l j_l s t} =& \sum_{\beta} | (n_{l \beta})_{i_l j_l s t} \rr\ll n_{l \beta})_{i_l j_l s t} |, \\
            \bar{Y}_{D_{m_l}}^* \coloneqq& \sum\limits_{\alpha \beta i_{0:l} j_{0:l} s t m_{1:l-1}} \bra{tj_0} \Cmat^{(0)}_{m_0} \ket{si_0}\ \prod_{r \neq l} \tr \big[ \Cmat_{m_{r+1}|m_r}^{(r+1)} \\
            &\;\times\left(N_r\right)_{i_rj_ri_{r+1}j_{r+1}} \big] (d_{m_l \alpha}^*)^\dag (n_{l \beta})_{i_lj_lst} \\
            &\;\times\tr[(n_{l \beta})_{i_lj_lst}^\dag d_{m_l \alpha_l}^*] ,\ \forall\ m_l.
        \end{aligned}
    \end{equation*}}
    The optimal dual solution $Y_{D_{m_l}}^*$ is non-degenerate, and the optimal primal solution $\Dmat_{m_l}^*$ is robust.
\end{lemma}
\begin{lemma}\label{lem:opt_Cmat_m_k}
    In line~\ref{eq:opt_Cmat_m_k} of Algorithm~\ref{algo:see-saw_gen}, $\Cmat_{m_{k+1}|m_k}^{(k+1)*} = \sum_{\gamma_{k+1}} |c_{m_{k+1} \gamma_{k+1}}^{(k+1)*} \rr\ll c_{m_{k+1} \gamma_{k+1}}^{(k+1)*}|$ are Choi matrices of the optimal check operation if and only if
    \begin{equation*}
        \left( I_{Q_{k+1}} \otimes Y^*_{C_{m_{k+1}|m_k}^{(k+1)*}} - (\Cmat_k \Dmat)_{m_{k+1}} \right) \Cmat^{(k+1)*}_{m_{k+1} | m_k} \geq 0\ \forall\ m_{k+1}
    \end{equation*}
    where
    \begin{equation*}
        \begin{aligned}
            (\Cmat_k \Dmat)_{m_{k+1}} =& \sum_{i_{0:l} j_{0:l} s t m_{1:k, k+2:l}} \bra{tj_0} \Cmat_{m_0}^{(0)} \ket{si_0} \\
            &\;\times\prod_{r \neq k, l} \tr [\Cmat_{m_{r+1} | m_r}^{(r+1)} (N_r)_{i_r j_r i_{r+1} j_{r+1}}]\ \\
            &\;\times\tr [D_{m_l} (N_l)_{i_l j_l s t}] (N_k)_{i_k j_k i_l k_l}, \\
            (N_k)_{i_k j_k i_l j_l} =& \sum_{\delta_l} | (n_{k \delta_l})_{i_k j_k i_l j_l} \rr\ll n_{k \delta_l})_{i_k j_k i_l j_l} |, \\
            \bar{Y}_{C_{m_{k+1} | m_k}}^{(k+1)*} \coloneqq &\sum_{\gamma_{k+1} \delta_l i_{0:l} j_{0:l} s t m_{1:k, k+2:l}} \bra{tj_0} \Cmat_{m_0}^{(0)} \ket{si_0} \\
            &\;\times \prod_{r \neq k, l} \tr [\Cmat_{m_{r+1} | m_r}^{(r+1)} (N_r)_{i_r j_r i_{r+1} j_{r+1}}] \\
            &\;\times\tr [\Dmat_{m_l} (N_l)_{i_l j_l s t}] (c_{m_{k+1} \gamma_{k+1}}^{(k+1) *})^\dag \\
            &\;\times(n_{k \delta_l})_{i_k j_k i_l j_l} \tr \left[ (n_{k \delta_l})^\dag_{i_k j_k i_l j_l} c_{m_{k+1} \gamma_{k+1}}^* \right].
        \end{aligned}
    \end{equation*}
    The optimal dual solution $Y_{C_{m_{k+1} | m_k}}^{(k+1)*}$ is non-degenerate, and the optimal primal solution $\Cmat_{m_{k+1}|m_k}^{(k+1)*}$ is robust.
\end{lemma}
\begin{lemma}\label{lem:opt_Cmat_m_0_gen}
    In line~\ref{eq:opt_Cmat_m_0l} of Algorithm~\ref{algo:see-saw_gen}, $\Cmat_{m_0}^{(0)*} = \sum_{\gamma_0} |c_{\gamma_0}^{(0)*} \rr\ll c_{\gamma_0}^{(0)*}|$ are Choi matrices of the optimal encoder if and only if
    \begin{equation*}
        (I_{Q_0} \otimes Y_{C_{m_0}}^* - \Cmat_l \Dmat) \Cmat_{m_0}^{(0)*} \geq 0
    \end{equation*}
    where
    \begin{equation*}
        \begin{aligned}
            \Cmat_l \Dmat =& \sum_{i_{0:l} j_{0:l} s t m_{1:l}} \prod_{r \neq l} \tr[\Cmat^{(r+1)}_{m_{r+1} | m_r} (N_r)_{i_r j_r i_{r+1} j_{r+1}}] \\
            &\;\times\tr[\Dmat_{m_l} (N_l)_{i_l j_l s t}] \ket{s i_0}\bra{t j_0}, \\
            \ket{si_0}\bra{tj_0} =& \sum_f | f \rr\ll f |, \\
            \bar{Y}_{C_{m_0}}^* \coloneqq &\sum_{\gamma_0 f i_{0:l} j_{0:l} s t m_{1:l}} \prod_{r \neq l} \tr \left[ \Cmat^{(r+1)}_{m_{r+1} | m_r} (N_r)_{i_r j_r i_{r+1} j_{r+1}} \right] \\
            &\;\times\tr \left[ \Dmat_{m_l} (N_l)_{i_L j_l s t} \right] (c_{\gamma_0}^{(0)*})^\dag f \tr \left[ f^\dag c_{\gamma_0}^{(0)*} \right].
        \end{aligned}
    \end{equation*}
    The optimal dual solution $Y_{C_{m_0}^*}$ is non-degenerate, and the optimal primal solution $\Cmat_{m_0}^{(0)*}$ is robust.
\end{lemma}

\section{\label{app:AD_noise}Amplitude Damping Noise Over Two Rounds}

A $k$-qubit local amplitude damping noise over $n$ physical qubits for $l=2$ rounds can be written as
\begin{equation}
\label{eq:AD_gen}
    \mathbf{E} = \sum_{K\in\mathcal Q_k}\sum_{i,j} \bigotimes_{k'\in K}\frac{1}{|\mathcal Q_k|}
    |E_{i,k'}^{(0)}\rr\ll E_{i,k'}^{(0)}|_{Q_0Q_0'} \otimes |E_{j,k'}^{(1)}\rr\ll E_{j,k'}^{(1)}|_{Q_1Q_1'}.
\end{equation}
where $\mathcal Q_k$ contains all possible combinations of $k$ qubits from $n$ physical qubits.
The Kraus operators are given by
\begin{equation}
\label{E:cor_noise}
    E_{0,*}^{(l)} = 
    \left[\begin{array}{cc}
        1 & 0 \\
        0 & \sqrt{1-\gamma_l}
    \end{array} \right],\;
    E_{1, *}^{(l)} = \left[\begin{array}{cc}
        0 & \sqrt{\gamma_l}\\
        0 & 0
    \end{array} \right],\,l\in\{0, 1\}
\end{equation}
with $\gamma=\gamma_0+\gamma_1-\gamma_0\gamma_1$ as the damping strength of the complete channel~\cite{utagi2020temporal}. For numerics, we have taken $\gamma_0=\gamma/2$ and $\gamma_1=\frac{\gamma-\gamma_0}{1-\gamma_0}=\frac{\gamma}{2-\gamma}$.

Also, the amplitude damping noise of weight $k$~\cite{cafaro2014approximate,jayashankar2022achieving,dutta2024smallest} over $n$ physical qubits for $l=2$ can be written as
\begin{equation}
\label{eq:new_noise}
\mathbf{E}=\sum_{i,j\leq k}|A^{(0)}_i\rr\ll A^{(0)}_i|\otimes|A^{(1)}_i\rr\ll A^{(1)}_j|,
\end{equation}
where the Kraus operators are given by
\begin{equation}
A^{(l)}_i=\bigotimes_{s=1}^nE^{(l)}_{i_s,*},
\end{equation}
such that $\sum\limits_{s=1}^ni_s=i$. Note that eq.~\eqref{eq:AD_gen} and~\eqref{eq:new_noise} give the same noise model when $k=n$.

\section{\label{app:low_weight}Strategic QECC Optimization with low-weight noise}

A low-weight noise is described by eq.~\eqref{eq:new_noise} when $k$ is small. Figure~\ref{fig:fidelity_plot_new} shows the plots of the entanglement fidelity $F_{ent}$ as a function of the damping strength $\gamma$ using Algorithm~\ref{algo:see-saw_k} for $k=1$ and $2$. This shows that the four-qubit code can improve fidelity effectively, surpassing the $\llbracket4, 1\rrbracket$ code from Ref.\cite{leung1997approximate}. Although the two- and three-qubit codes do not guarantee an improvement over the no-encoding case, they outperform the $\llbracket3, 1\rrbracket$ code from Ref.\cite{dutta2024smallest} with the instances where the encoders are combined with Petz recovery map decoders.

\begin{figure*}
\centering
\includegraphics[width=0.475\textwidth]{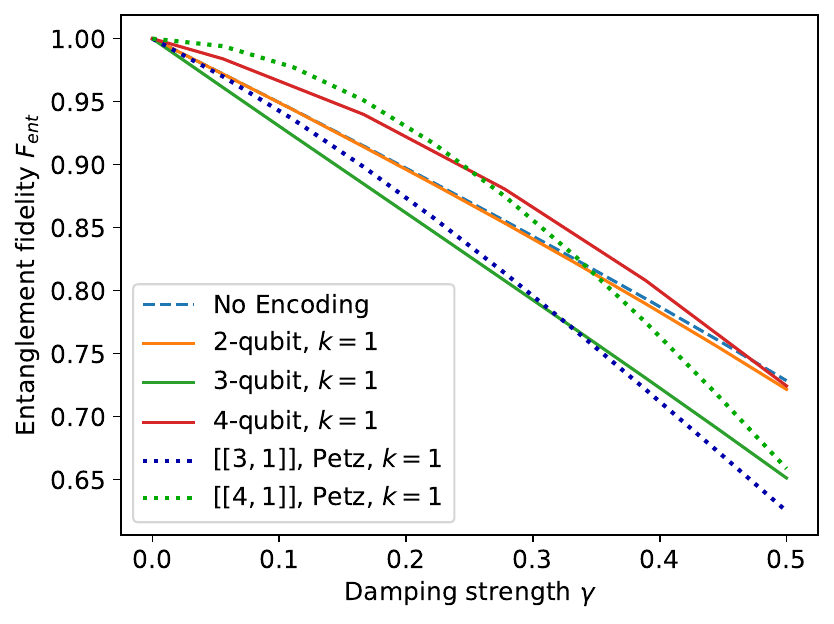}\hfill\includegraphics[width=0.475\textwidth]{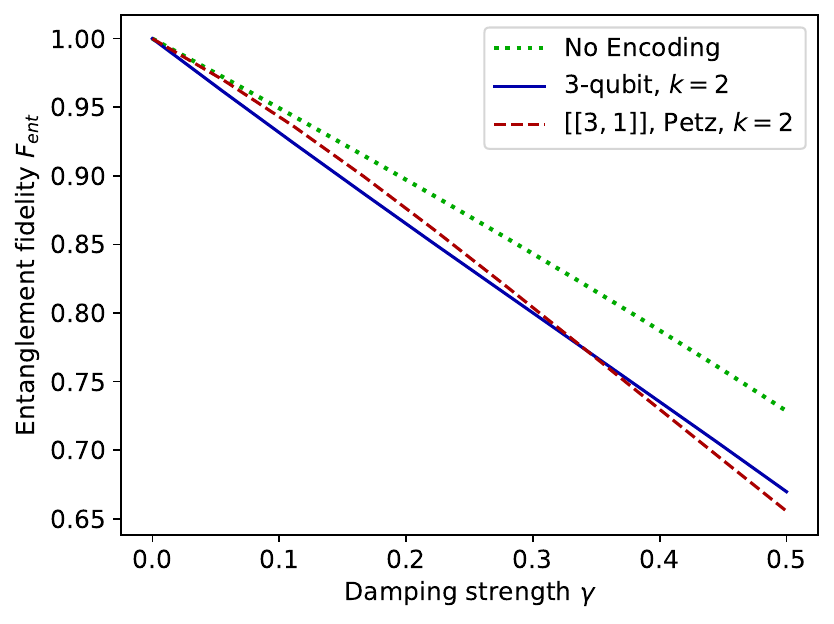}
\caption{\label{fig:fidelity_plot_new}
The entanglement fidelity $F_{ent}$ versus the damping strength $\gamma$ is plotted using Algorithm~\ref{algo:see-saw_k}
for the noise model given by eq.~\eqref{eq:new_noise} with $k=1$ (left) and $k=2$ (right). A single check operator, i.e., $m=1$ has been considered in all cases.
Observe that our four-qubit code can improve the fidelity effectively, surpassing the $\llbracket4, 1\rrbracket$ code from Ref.~\cite{leung1997approximate}. Although the two- and three-qubit codes can't ensure any improvement over no encoding case, they outperform the $\llbracket3, 1\rrbracket$ code from Ref.~\cite{dutta2024smallest}, along with the cases where these encoders are paired with the Petz recovery map decoders, in the presence of all-qubit errors.
}
\end{figure*}

\section{\label{app:all_error}Strategic QECC Optimization with All-Error Qubit}

The entanglement fidelity $F_{ent}$ as a function of the damping strength $\gamma$ is plotted in Fig.~\ref{fig:fidelity_plot_all} using Algorithm~\ref{algo:see-saw_k} for $2$- and $3$-qubit codes with a single check operator ($m=1$) in the all-qubit noise model. Although these codes do not show any improvement over the unencoded case, they outperform the $\llbracket5,1\rrbracket$ and $\llbracket3,1\rrbracket$ codes from Ref.\cite{dutta2024smallest}, as well as the cases where these codes are combined with Petz recovery map decoders, in the presence of all-qubit errors. However, for small damping strength, the $\llbracket4,1\rrbracket$ code combined with Petz recovery map decoders from Ref.\cite{leung1997approximate} performs better than our codes.

\begin{figure*}[htpb]
\centering
\includegraphics[width=.9\textwidth]{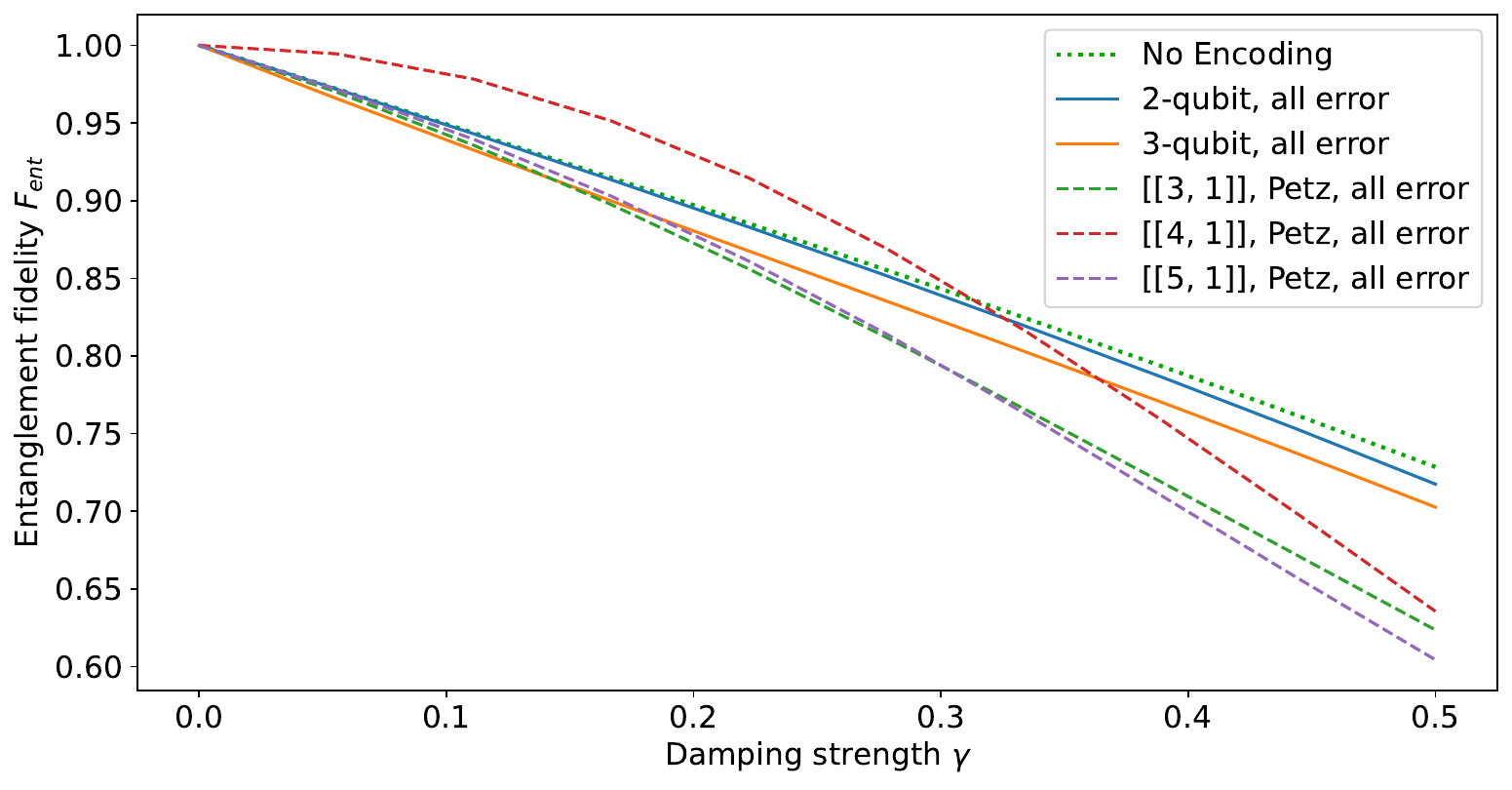}
\caption{\label{fig:fidelity_plot_all}
The entanglement fidelity $F_{ent}$ versus the damping strength $\gamma$ is plotted using Algorithm~\ref{algo:see-saw_k}
for $2$- and $3$--qubit codes with a single check operator, i.e., $m=1$, in the all-qubit noise scenario.
Observe that although these codes can't ensure any improvement over no encoding case, they outperform the $\llbracket5, 1\rrbracket$ and $\llbracket3, 1\rrbracket$ codes from Ref.~\cite{dutta2024smallest}, along with the cases where these encoders are paired with the Petz recovery map decoders, in the presence of all-qubit errors. However, for small damping strength, the $\llbracket4, 1\rrbracket$ code combined with Petz recovery map decoders from Ref.~\cite{leung1997approximate} performs better than our results.
}
\end{figure*}

\section{Proof of Theorem~\ref{thm:petz_rec}}
\label{app:proof_petz_recovery_optimal_for_correctable_noise}

Here we give a proof of Theorem~\ref{thm:petz_rec},
which we restate as follows, for the reader's convenience:
If strategic code $\mathbf{C}$ corrects noise $\mathbf{E}$, then the dynamical Petz recovery map $\mathcal{R}_{\mathbf{C},\mathbf{E}}^\mathrm{Petz}$ is equal to the perfect recovery map $\mathcal{R}_{\mathbf{C},\mathbf{E}}^*$.
This perfect recovery map is defined by a set of recovery channels $\{\mathcal{R}_m^*\}_m$ that perfectly recovers the logical information which can be constructed whenever the condition in Fact~\ref{thm:algebraic_KL_condition}
is satisfied.
Each recovery map $\mathcal{R}_m^*$ has Kraus operators of the form
\begin{equation}\label{eqn:kraus_perfect_dynamic_recovery}
\begin{aligned}
    R_{e|m} = \frac{1}{\sqrt{d_{e',e,m}}} \sum_{o \in O_m} \ll F_{e}| (|C_{m,o}\rr |\Pi_{Q_0}\rr) \;.
\end{aligned}
\end{equation}
Here $F_e = \sum_{e'} u_{e',e} E_{e'}$ (hence, $|F_e\rr = \sum_{e'} u_{e',e} |E_{e'}\rr$) where $[u_{e',e}]_{e',e}$ is a unitary matrix such that $d_{e',e,m} = \sum_{\Tilde{e},\Bar{e}} u_{e',\Tilde{e}}^* \lambda_{\Tilde{e},\Bar{e},m} u_{\Bar{e},e}$ with $d_{e,e',m}=0$ if $e\neq e'$.
The idea of the proof is to show that the Kraus operators of the Petz recovery coincide with the Kraus operators of the perfect recovery, for each $m$.

\begin{proof}
    First, we rewrite the necessary and sufficient condition in Fact~\ref{thm:algebraic_KL_condition}
    as 
    \begin{equation}
        (\ll C_m|\otimes\Pi_{Q_0}) E_{e'}^\dag E_e (|C_m\rr\otimes\Pi_{Q_0}) = c_{e',e,m} \Pi_{Q_0}
    \end{equation}
    Equivalently, we can also rewrite it as
    \begin{equation}\label{eqn:diag_dynamical_condition}
        (\ll C_m|\otimes\Pi_{Q_0}) F_{e'}^\dag F_e (|C_m\rr\otimes\Pi_{Q_0}) = \delta_{e',e} d_{e,m} \Pi_{Q_0} \;,
    \end{equation}
    by diagonalizing $[c_{e',e,m}]_{e',e} \mapsto \mathrm{diag}\{d_{e,m}\}_e$ such that $d_{e,m} = \sum_{e',e''} u_{e,e'}^* c_{e',e'',m} u_{e'',e}$ and $F_e = \sum_{e'} u_{e',e} E_{e'}$ for some unitary matrix $[u_{e',e}]_{e',e}$.
    Now consider polar decomposition $F_e (|C_m\rr \otimes \Pi_{Q_0}) = U_{e,m} \Pi_{Q_0} \sqrt{d_{e,m}}$ for unitary $U_{e,m}$ and a real number $d_{e,m}$.
    Thus we can rewrite eqn.~\eqref{eqn:diag_dynamical_condition} as
    \begin{equation}\label{eqn:polar_decomp_dynamical_condition}
    \begin{gathered}
        \sqrt{d_{e',m}d_{e,m}} \Pi_{Q_0} U_{e',m}^\dag U_{e,m} \Pi_{Q_0} = \delta_{e',e} d_{e,m} \Pi_{Q_0} \\
        \Rightarrow \Pi_{Q_0} U_{e',m}^\dag U_{e,m} \Pi_{Q_0} = \delta_{e',e} \Pi_{Q_0} \;.
    \end{gathered}
    \end{equation}
    For a channel $\mathcal{R}_m$ with Kraus operators $\{\Pi_{Q_0}U_{e,m}^\dag\}_e$, we have
    {\allowdisplaybreaks
    \begin{align}
        &\mathcal{R}_m \circ \mathbf{E}\big(|C_m\rr\ll C_m| \otimes \Pi_{Q_0}\big)\notag \\
        &= \sum_{e',e} \Pi_{Q_0}U_{e',m}^\dag F_e |C_m\rr\ll C_m| \otimes \Pi_{Q_0} F_e^\dag U_{e',m} \Pi_{Q_0}\notag \\
        &= \sum_{e',e} d_{e,m} (\ll C_m|\otimes\Pi_{Q_0}) F_{e'}^\dag F_e (|C_m\rr\ll C_m| \otimes \Pi_{Q_0}) F_e^\dag\notag \\
        &\quad F_{e'} (|C_m\rr\otimes\Pi_{Q_0})\notag \\
        &= \sum_{e',e} d_{e,m} (\delta_{e',e} d_{e,m} \Pi_{Q_0}) (\delta_{e',e} d_{e,m} \Pi_{Q_0})\notag \\
        &= \Big(\sum_e d_{e,m}^3\Big) \Pi_{Q_0}
    \end{align}}
    where the second equality uses the polar decomposition above and the third equality is by eqn.~\eqref{eqn:diag_dynamical_condition}.
    Hence recovery map $\mathcal{R}_m$ is the perfect recovery map $\mathcal{R}_{\mathbf{C},\mathbf{E}}^*$ as $\mathcal{R}_m \circ \mathbf{E}\big(|C_m\rr\ll C_m| \otimes \Pi_{Q_0}\big) \propto \Pi_{Q_0}$.
    Note that the Kraus operators of recovery channel in eqn.~\eqref{eqn:kraus_perfect_dynamic_recovery} is precisely the channel $\mathcal{R}_m$ with Kraus operators $\{\Pi_{Q_0}U_{e,m}^\dag\}_e$, as
    \begin{equation}
    \begin{aligned}
        R_{e|m} &= \frac{1}{\sqrt{d_{e,m}}} \ll F_e| (|C_m\rr\otimes|\Pi_{Q_0}\rr) \\
        &= \frac{1}{\sqrt{d_{e,m}}} (\ll C_m| \otimes \Pi_{Q_0}) F_e^\dag \\
        &= \Pi_{Q_0} U_{e,m}^\dag
    \end{aligned}
    \end{equation}
    by applying the polar decomposition above at the last equality.

    Now note that by again using the polar decomposition we have
    \begin{equation}
    \begin{aligned}
        &\mathbf{E}(|C_m\rr\ll C_m| \otimes \Pi_{Q_0}) \\
        &= \sum_e F_e \big( |C_m\rr\ll C_m| \otimes \Pi_{Q_0} \big) F_e^\dag \\
        &= \sum_e d_{e,m} U_{e,m} \Pi_{Q_0} \Pi_{Q_0} U_{e,m}^\dag \\
        &= \sum_e d_{e,m} \Pi_{e,m}
    \end{aligned}
    \end{equation}
    where $\Pi_{e,m} = U_{e,m}\Pi_{Q_0}U_{e,m}^\dag$ and $\Pi_{e,m}\Pi_{e',m}=\delta_{e,e'}\Pi_{e,m}$ (i.e. $\{\Pi_{e,m}\}_e$ is a set of orthogonal projectors, due to eqn.~\eqref{eqn:polar_decomp_dynamical_condition}).
    Hence we obtain
    \begin{equation}
    \begin{aligned}
        \big(\mathbf{E}(|C_m\rr\ll C_m| \otimes \Pi_{Q_0})\big)^{-\frac{1}{2}} = \sum_e \Pi_{e,m} (d_{e,m})^{-\frac{1}{2}} \;,
    \end{aligned}
    \end{equation}
    which allows us to rewrite the Kraus operators of the Petz recovery map as
    \begin{equation}
    \begin{aligned}
        R_{e|m} &= \sum_{e'} (\ll C_m|\otimes\Pi_{Q_0}) F_e^\dag \Pi_{e',m} (d_{e',m})^{-\frac{1}{2}} \\
        &= \sum_{e'} \sqrt{d_{e,m}}  \Pi_{Q_0}U_{e,m}^\dag \Pi_{e',m} (d_{e',m})^{-\frac{1}{2}} \\
        &= \sum_{e'} \Pi_{Q_0} U_{e,m}^\dag U_{e',m}\Pi_{Q_0}U_{e',m}^\dag \\
        &= \Pi_{Q_0} U_{e,m}^\dag \;,
    \end{aligned}
    \end{equation}
    by using eqn.~\eqref{eqn:polar_decomp_dynamical_condition}.
    Hence we have shown that the Kraus operators of the Petz recovery map $\mathcal{R}_{\mathbf{C},\mathbf{E}}^\mathrm{Petz}$ is equal to the Kraus operators of the perfect recovery map $\mathcal{R}_{\mathbf{C},\mathbf{E}}^*$.
\end{proof}

\section{Fidelity Bound for Approximate Strategic Code}

Here we describe the scenario considered in Fact~\ref{thm:info_theoretic_QECC}
 and Theorem~\ref{thm:info_theoretic_AQECC},
then give a proof for Theorem~\ref{thm:info_theoretic_AQECC}.

\subsection{Purified dynamics of strategic code and noise}\label{app:strategic_code_details:purified}

First, we explicitly derive the purification of the interrogator $\mathbf{T}_m$ given a memory trajectory $m$ and noise $\mathbf{E}$ which gives a purified dynamics $\mathbf{E}\ast\mathbf{T}_m$ that maps input maximally entangled state $|\phi\>$ to a global pure state $|\mathbf{E}\ast\mathbf{T}_m(\phi)\>$.
An illustration of this scenario for a three-round strategic code is shown in Fig.~\ref{fig:purified_strategic_code}.
Consider the quantum channels $\mathcal{C}^{(1)},\mathcal{C}_{m_1}^{(2)},\mathcal{C}_{m_{1:3}}^{(3)},\dots$ corresponding to operations performed by the interrogator $\mathbf{T}_m$ for a given memory trajectory $m=m_1,m_2,\dots$.
Each of these quantum channel admits a Kraus operator representation $\{C_{m_r|m_{1:r}}^{(r)}\}_{m_r}$ (see Appendix~\ref{app:strategic_code_details}).
We can use these Kraus operators to construct a Stinespring representation of these channels with isometry $V_{m_{1:r}}:= \sum_{m_r} C_{m_r|m_{1:r}}^{(r)}\otimes|m_r\>_{B_r'}$ so that $\mathcal{C}_{m_{1:r}}^{(r)}(\rho) = \tr_{B_r'}(V_{m_{1:r}} \rho V_{m_{1:r}}^\dag)$, where $B_r'$ is an ancillary system with basis $\{|m_r\>\}_{m_r}$. 
For trajectory $m=m_1,\dots,m_l$ we can now define
\begin{equation}
    |\Tilde{\mathbf{T}}_m(\phi)\rr := \Big(\bigotimes_{r=1}^l |C_{m_r|m_{1:r}}^{(r)}\rr\otimes|m_r\>_{B_r}\Big) \otimes |\phi\> \;.
\end{equation}

\begin{figure*}
    \centering
    \includegraphics[scale=0.35]{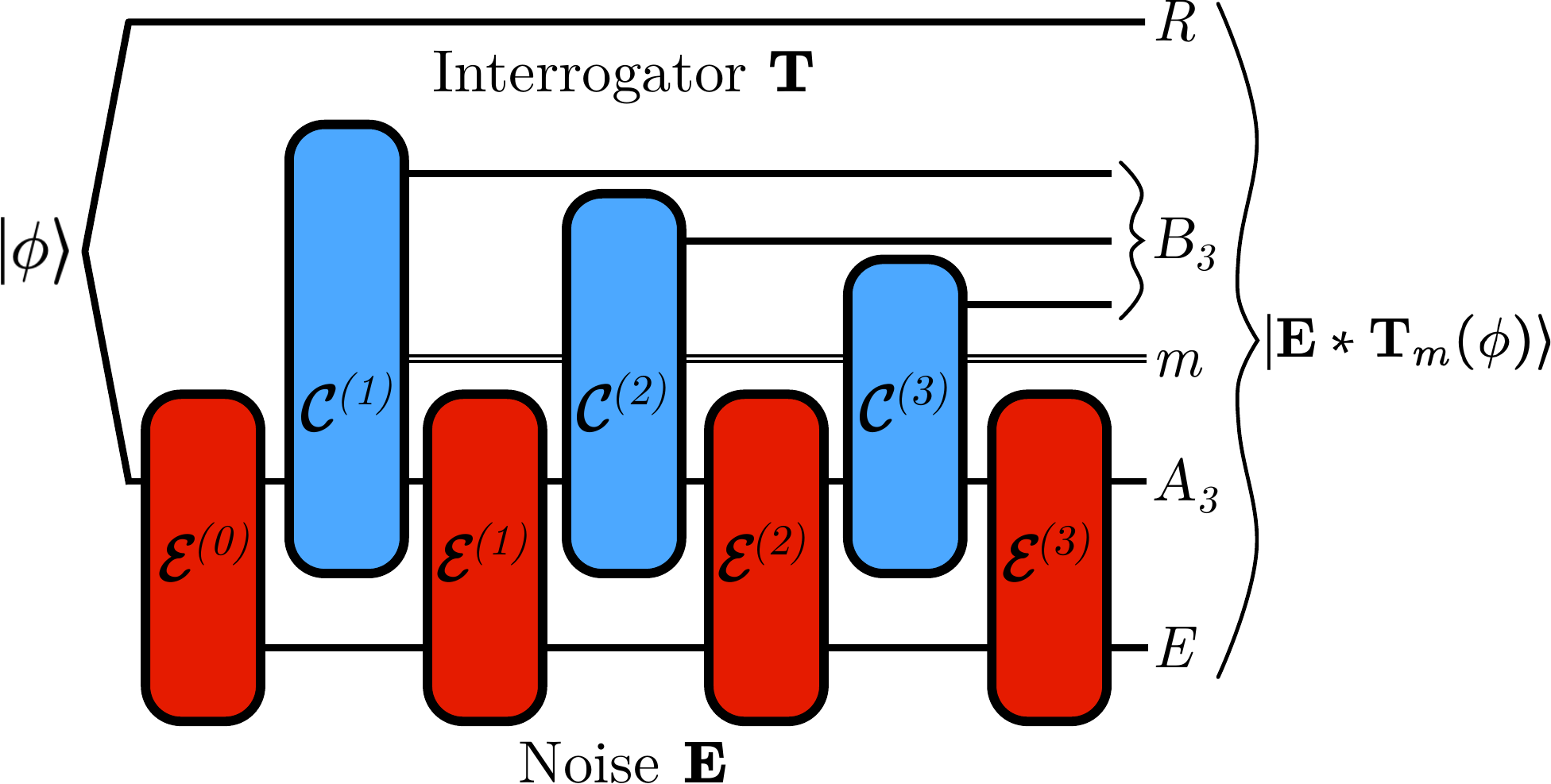}
    \caption{Transmission of one-half of maximally entangled state $|\phi\>$ through a sequence of three intermediary operations $\mathcal{C}^{(1)},\mathcal{C}^{(2)},\mathcal{C}^{(3)}$ interacting with noise $\mathbf{E}$.
    Purification of $\mathcal{C}^{(1)},\mathcal{C}^{(2)},\mathcal{C}^{(3)}$ and $\mathbf{E}$ result in ancillary system $B_3$ and noise environment $E$.
    Global state at the output over systems $R,B_3,A_3,E$ for a given memory $m$ is given by a pure state $|\mathbf{E}\ast\mathbf{T}_m(\phi)\>$. 
    }
    \label{fig:purified_strategic_code}
\end{figure*}

Purification of noise $\mathbf{E}$ can be constructed similarly.
At the start of the decoding stage, we have a global state $\sum_m\mathbf{E}(|\Tilde{\mathbf{T}}_m(\phi)\rr\ll\Tilde{\mathbf{T}}_m(\phi)|)$ where map $\mathbf{E}$ is acting trivially on systems $B_l=B_1'\otimes\dots\otimes B_l'$ and $R$, where $|E_e\rr = \sum_{j} \alpha_{j} |j_0\>\otimes|j_0'\>\otimes\dots\otimes|j_l\>\otimes|j_l'\>$ with $\{|j_r\>\}_{j_r=1}^{d_r}$ (see Appendix~\ref{app:strategic_code_details}).
Noise $\mathbf{E} = \sum_e |E_e\rr\ll E_e|$ can be thought of as a quantum channel with Kraus operators $\mathbf{E}_e = \sum_{j} \alpha_{j} |j_l'\> (\<j_0|\otimes\<j_0'|\otimes\dots\otimes\<j_l|)$ which gives
\begin{equation}
    \mathbf{E}(|\Tilde{\mathbf{T}}_m(\phi)\rr\ll\Tilde{\mathbf{T}}_m(\phi)|) = \mathbf{E}_e |\Tilde{\mathbf{T}}_m(\phi)\rr\ll\Tilde{\mathbf{T}}_m(\phi)|\mathbf{E}_e^\dag \;.
\end{equation}
Since $\mathbf{E}$ is a quantum channel, it admits a Stinespring representation $\mathbf{E}(A) = \tr_E(V_\mathbf{E} A V_\mathbf{E}^\dag)$ for isometry $V_\mathbf{E} = \sum_e \mathbf{E}_e\otimes|e\>_E$.
Isometry $V_\mathbf{E}$ is the purified noise $\mathbf{E}$ to some noise environment $E$ with basis $\{|e\>\}_e$.

By using $|\Tilde{\mathbf{T}}_m(\phi)\rr$ and $V_\mathbf{E}$, we obtain the global purified dynamics which gives us the global state at the start of the decoding round $|\mathbf{E}\ast\mathbf{T}_m(\phi)\>:= V_\mathbf{E}|\Tilde{\mathbf{T}}_m(\phi)\rr$, which can be expressed (with implicit ``$\otimes$'') as:
\begin{equation}
    |\mathbf{E}\ast\mathbf{T}_m(\phi)\> = \sum_{i,e} |\Bar{i}\>_R \big(K_{e,m}|\Bar{i}\>_{A_0}\big) |m\>_{B_l} |e\>_E
\end{equation}
where $|m\>_{B_l} = |m_1\>_{B_1'}\dots|m_l\>_{B_l'}$.
Linear operator $K_{e,m} = \mathbf{E}_e(\bigotimes_{r=1}^l|C_{m_r|m_{1:r}}^{(r)}\rr)$ is to be interpreted as the Kraus operator (over indices $e$) for the transformation $A_0\mapsto A_l$ induced by the purified interaction between $\Tilde{\mathbf{T}}_m$ and $V_\mathbf{E}$.
The reduces states of the global state $|\mathbf{E}\ast\mathbf{T}_m(\phi)\>$ to system $W$ is simply given by $\rho_m^W := \tr_{\backslash W}(|\mathbf{E}\ast\mathbf{T}_m(\phi)\>\<\mathbf{E}\ast\mathbf{T}_m(\phi)|)$ where $W$ is a subset of systems $R,B_l,A_l,E$ and $\tr_{\backslash W}$ indicates partial trace over systems not in $W$, e.g. $\rho_m^{RE} = \tr_{B_lA_l}(|\mathbf{E}\ast\mathbf{T}_m(\phi)\>\<\mathbf{E}\ast\mathbf{T}_m(\phi)|)$

\subsection{Proof of Theorem~\ref{thm:info_theoretic_AQECC}}
\label{app:info_theoretic_AQECC}

Theorem~\ref{thm:info_theoretic_AQECC}
states that if $S(\rho_m^{R}) + S(\rho_m^{B_l E}) - S(\rho_m^{R B_l E}) < \varepsilon$ for each value of $m$, then then the entanglement fidelity between the recovered codestate $\sigma^{R Q_0}$ and the initial codestate $\ket{\phi}_{R Q_0}$ is bounded as $$F_{ent}(\sigma^{R Q_0}, \ket{\phi}_{R Q_0}) \geq 1 - 2\sqrt{\varepsilon}.$$

\begin{proof}
Note that
\begin{equation*}
    \begin{aligned}
        \varepsilon &> S(\rho_m^{R}) + S(\rho_m^{B_l E}) - S(\rho_m^{R B_l E}) \\
        &= S(\rho_m^{R B_l E} \| \rho_m^{R} \otimes \rho_m^{B_l E}),
    \end{aligned}
\end{equation*}
where $S(\rho \| \sigma) = \tr(\rho \log(\rho)) - \tr(\rho \log(\sigma))$ denotes the relative entropy between the states $\rho$ and $\sigma$.

Ref.~\cite{schumacher2002approximate} states that fidelity between any states $\rho$ and $\sigma$ can be bounded by
\begin{equation*}
    F(\rho, \sigma) \geq 1 - \sqrt{S(\rho \| \sigma)}.
\end{equation*}

Thus, we have
\begin{equation*}
    \begin{aligned}
        F(\rho_m^{R B_l E}, \rho_m^{R} \otimes \rho_m^{B_l E}) &\geq 1 - \sqrt{S(\rho_m^{R B_l E} \| \rho_m^{R} \otimes \rho_m^{B_l E})}\\
        &> 1 - \sqrt{\varepsilon}.
    \end{aligned}
\end{equation*}

The state $\ket{\psi_m^{R A_l B_l E}}$ of the system $R A_l B_l E$ is a purification of $\rho_m^{R B_l E}$.
Let $\ket{\xi_m^{R A_l B_l E}}$ be a specific purification of $\rho_m^{R} \otimes \rho_m^{B_l E}$ satisfying
\begin{equation*}
    F(\rho_m^{R B_l E}, \rho_m^{R} \otimes \rho_m^{B_l E}) = \left| \left\langle \psi_m^{R A_l B_l E}\ |\ \xi_m^{R A_l B_l E} \right\rangle \right|.
\end{equation*} 

The Schmidt decomposition of the purification $\ket{\xi_{m}^{R A_l B_l E}}$ on bipartition $A_l:R B_l E$ can be expressed as
\begin{equation*}
    \ket{\xi_{m}^{R A_l B_l E}} = \sum_{i, \alpha} \sqrt{q_\alpha^{(m)}} \ket{i}_{R} \ket{u_{\alpha}^{(m)}}_{B_l E} \ket{v_{i, \alpha}^{(m)}}_{A_l},
\end{equation*}
where $\left\{ \ket{u_\alpha^{(m)}}_{B_l E} \right\}_\alpha$ is an eigenvector of $\rho_{B_l E}$ with corresponding eigenvalue $q_\alpha^{(m)}$ and $\left\{ \ket{v_{i, \alpha}^{(m)}}_{A_l} \right\}_{i, \alpha}$ is a basis of $A_l$.

Observe that this is the same as Eq.~(29) of~\cite{tanggara2024strategic} and therefore we can apply the decoding channel $\mathcal{D}_{m}$ (described in the proof of Fact~\ref{thm:info_theoretic_QECC}
in~\cite{tanggara2024strategic}) to restore $\ket{\xi_{m}^{R A_l B_l E}}$ to the initial codestate which is one half of $\ket{\phi}_{R Q_0}$.

However, instead of the purification $\ket{\xi_{m}^{R A_l B_l E}}$ of $\rho_{m}^{R} \otimes \rho_{m}^{B_l E}$, we have the purification $\ket{\psi_{m}^{R A_l B_l E}}$ of $\rho_{m}^{R B_l E}$.
Upon applying the recovery channel $\mathcal{D}_{m}$ to $\ket{\psi_{m}^{R A_l B_l E}}$, we obtain the state $\sigma^{R A_0}$.
Since fidelity cannot be decreased by any operation, we have
\begin{equation*}
    \begin{aligned}
        F(\sigma^{R A_0}, \ket{\phi}_{R Q_0}) &\geq \left| \left\langle \psi_{m}^{R A_l B_l E} | \xi_{m}^{R A_l B_l E} \right\rangle \right| \\
        &= F(\rho_{m}^{R B_l E}, \rho_{m}^{R} \otimes \rho_{m}^{B_l E}) \\
        &> 1 - \sqrt{\epsilon}.
    \end{aligned}
\end{equation*}

Since the entanglement fidelity $F_{ent}$ is related to the fidelity $F$ as $F_{ent} = F^2$, we have
\begin{equation}
    F_{ent}(\sigma^{R A_0}, \ket{\phi}_{R Q_0}) > (1 - \sqrt{\varepsilon})^2 \geq 1 - 2\sqrt{\varepsilon}.
\end{equation}
\end{proof}

\section{\label{app:obj}Alternate Expression of the Objective Function}

Here, we derive an alternative expression for the Entanglement fidelity as the objective function of the semidefinite program, allowing for separate optimization of different operators, namely the encoder, interrogator, and decoder.
First, we rewrite the noise and input state as
{\allowdisplaybreaks
\begin{align*}
\mathbf{E}^{(l)}&=\sum_{i,j}|i\rangle\langle j|_{A_l}\otimes\mathcal{E}^{(l)}\left(|i\rangle\langle j|_{A_l}\right)\\
&=\sum_{i,j}|i\rangle\langle j|_{A_l}\otimes\left(\sum_{e_l}E_{e_l}|i\rangle\langle j|E_{e_l}^\dagger\right)_{A_l} \\
\quad\text{ and}\;
|\rho\rr\ll\rho|&=\sum_{i,j}|i\rangle\langle j|_{L}\otimes\left(\rho|i\rangle\langle j|\rho^\dagger\right)_{L'}.
\end{align*}}
Then we can write the fidelity Eq.~\eqref{eqn:entanglement_fidelity} 
as
\begin{widetext}
\allowdisplaybreaks
\begin{align*}
&\tr\Big[\left(I_L\otimes\sum_{i,j}|i\rangle\langle j|_{Q_0}\otimes\left(\sum_{e_0}E_{e_0}^*|i\rangle\langle j|E_{e_0}^\top\right)_{Q_0'}\otimes\sum_{h,k}|h\rangle\langle k|_{Q_1}\otimes\left(\sum_{e_1}E_{e_1}^*|h\rangle\langle k|E_{e_1}^\top\right)_{Q_1'}\otimes I_{L'}\right)\\
&\times\left(\sum_m\mathbf{C}^{(0)}_{L,Q_0}\otimes(\mathbf{C}_m^{(1)})_{Q_0',Q_1}\otimes(\mathbf{D}_m)_{Q_1',L'}\right)\times\left(\sum_{s,t}|s\rangle\langle t|_L\otimes I_{Q_0,Q_0',Q_1,Q_1'}\otimes\left(\rho|s\rangle\langle t|\rho^\dagger\right)_{L'}\right)\Big]\\
&=\sum_{m,i,j,h,k,s,t}\tr\Big[\left(I_L\otimes|i\rangle\langle j|\right)\mathbf{C}^{(0)}\left(|s\rangle\langle t|_L\otimes I_{Q_0}\right)\Big]\tr\Big[\left(\left(\sum_{e_0}E_{e_0}^*|i\rangle\langle j|E_{e_0}^\top\right)_{Q_0'}\otimes|h\rangle\langle k|_{Q_1}\right)\mathbf{C}^{(1)}_mI_{Q_0',Q_1}\Big]\\
&\tr\Big[\left(\left(\sum_{e_1}E_{e_1}^*|h\rangle\langle k|E_{e_1}^\top\right)_{Q_1'}\otimes I_{L'}\right)\mathbf{D}_m\left(I_{Q_1'}\otimes\left(\rho|s\rangle\langle t|\rho^\dagger\right)_{L'}\right)\Big]\\
&= \sum_{mijhkst} \tr \left[ \Cmat^{(0)} \left( \ket{s}\bra{t} \otimes \ket{i}\bra{j} \right) \right]\ \tr \left[ \Cmat_m^{(1)} \left( \left( \sum_{e_0} E_{e_0}^* \ket{i}\bra{j} E_{e_0}^\top \right) \otimes \ket{h} \bra{k} \right) \right]\ \tr \left[ \Dmat_m \left( \left( \sum_{e_1} E_{e_1}^* \ket{h}\bra{k} E_{e_1}^\top \right) \otimes (\rho \ket{s} \bra{t} \rho^\dag) \right) \right] \\
&= \sum_{mijhkst} \bra{tj} \mathbf{C}^{(0)} \ket{si}\ \tr \left[ \Cmat_m^{(1)} \left( \left( \sum_{e_0} E_{e_0}^* \ket{i}\bra{j} E_{e_0}^\top \right) \otimes \ket{h} \bra{k} \right) \right]\ \tr \left[ \Dmat_m \left( \left( \sum_{e_1} E_{e_1}^* \ket{h}\bra{k} E_{e_1}^\top \right) \otimes (\rho \ket{s} \bra{t} \rho^\dag) \right) \right].
\end{align*}
\end{widetext}

\end{document}